\documentclass[a4paper,onecolumn, accepted = 2025-04-21]{quantumarticle}

\pdfoutput = 1

\usepackage[utf8]{inputenc}

\usepackage{xcolor}
\definecolor{blueviolet}{rgb}{0.2, 0.2, 0.6}
\definecolor{webgreen}{rgb}{0,.5,0}
\definecolor{webbrown}{rgb}{.6,0,0}
\usepackage[pdftex,
  bookmarks=false,
  colorlinks=true, 
  urlcolor=webbrown,
  linkcolor=blueviolet, 
  citecolor=webgreen,
  pdfstartpage=1,
  pdfstartview={FitH},  
  bookmarksopen=false
  ]{hyperref}
\usepackage{amsmath}
\usepackage{amssymb}
\usepackage{amsfonts}
\usepackage{cleveref}

\usepackage{enumerate}
\usepackage{comment}
\usepackage{xpatch}
\usepackage{graphicx}
\usepackage{tabularx}
\usepackage{braket}
\usepackage{physics}
\usepackage{amsthm}
\usepackage{tikz}
\usepackage{commath}
\usepackage{mathtools}
\usepackage{qcircuit}
\usepackage{algorithm}
\usepackage{algorithmicx}
\usepackage{algpseudocode}
\usepackage{tabto}
\usepackage{pgf-umlsd}
\usepackage{bm}
\usepackage{multirow, multicol}
\usepackage{float}
\usepackage{pdfpages}
\usepackage{caption}
\usepackage{subcaption}
\usepackage{makecell}
\usepackage[numbers, sort&compress]{natbib}

\newcommand{\mc}[1]{\mathcal{#1}}

\newcommand{\dtr}{\mathrm{d_{tr}}}

\newcommand{\pauli}{\mathcal{P}}

\newcommand{\phidep}{\Phi_{\mathrm{dep}}}
\newcommand{\stabone}{\mathrm{stab}_1}
\newcommand{\stabn}{\mathrm{stab}_1^{\otimes n}}

\makeatletter
\newenvironment{protocol}[1][htb]{%
  \renewcommand{\ALG@name}{Protocol}
  \begin{algorithm}[#1]%
  }{\end{algorithm}
}

\newtheorem{theorem}{Theorem}
\newtheorem{prop}{Proposition}
\newtheorem{lemma}{Lemma}
\newtheorem{corollary}{Corollary}

\newtheorem{definition}{Definition}
\newtheorem{problem}{Problem}

\newtheorem{assumption}{Assumption}
\theoremstyle{remark}
\newtheorem{remark}{Remark}

\hypersetup{
    colorlinks=true,
    linkcolor=blue,
    filecolor=magenta,      
    urlcolor=cyan,
    pdfpagemode=FullScreen,
    citecolor=blue
}
\definecolor{todocol}{rgb}{1,0,0.0}
\definecolor{fixmecol}{rgb}{1.0,0.3,0.3}
\definecolor{ideacol}{rgb}{1.0,0.75,0.0}
\definecolor{probcol}{rgb}{1.0,0.1,0.0}
\definecolor{green-munsell}{rgb}{0.0, 0.66, 0.47}

\title{Learning Quantum Processes with Quantum Statistical Queries}

\author{Chirag Wadhwa}
\affiliation{School of Informatics, University of Edinburgh, Edinburgh, United Kingdom}
\email{chirag.wadhwa@ed.ac.uk}

\author{Mina Doosti}
\affiliation{School of Informatics, University of Edinburgh, Edinburgh, United Kingdom}
\email{mdoosti@ed.ac.uk}

\date{}

\begin{document}

\maketitle

\begin{abstract}
In this work, we initiate the study of learning quantum processes from quantum statistical queries. We focus on two fundamental learning tasks in this new access model: shadow tomography of quantum processes and process tomography with respect to diamond distance. For the former, we present an efficient average-case algorithm along with a nearly matching lower bound with respect to the number of observables to be predicted. For the latter, we present average-case query complexity lower bounds for learning classes of unitaries. We obtain an exponential lower bound for learning unitary 2-designs and a doubly exponential lower bound for Haar-random unitaries. Finally, we demonstrate the practical relevance of our access model by applying our learning algorithm to attack an authentication protocol using Classical-Readout Quantum Physically Unclonable Functions, partially addressing an important open question in quantum hardware security. 
\end{abstract}

\section{Introduction}
\label{sec:intro}
Quantum learning theory aims to study the advantages and limitations of quantum machine learning in both classical and quantum problems. Here, the \emph{learner} is a classical, quantum or hybrid algorithm, trying to learn about an unknown object (function, distribution, quantum state, quantum process, etc.) through some limited access to it. Widely studied classical access models such as random examples \cite{valiant1984theory} and statistical queries \cite{SQLearning} have been extended to the quantum setting in the form of quantum examples~\cite{QPAC,arunachalam2018optimal} and quantum statistical queries~\cite{arunachalam2020quantum, arunachalam2023role, hinsche2023one, nietner2023average, nietner2023unifying} respectively.  A wide array of results have been shown in this field, ranging from learning functions encoded within quantum states~\cite{QPAC,atici2005improved,grilolearning}, quantum state tomography~\cite{bisio2009optimal,o2016efficient,xu2018neural}, shadow tomography~\cite{classical_shadow_tomography,aaronson2018shadow}, learning diverse classes of probability distributions~\cite{childs2022quantum,atici2005improved,montanaro2017learning}, to learning quantum processes~\cite{mohseni2008quantum,chung2018sample,haah2023query}. We refer to the following survey for an overview of results~\cite{arunachalam2017guest}. While most of the efforts in quantum learning theory have been focused on quantum states (both as examples and target objects), in this work, we study a new access model for learning about quantum processes.

Learning quantum processes is a fundamental problem that arises in many areas in physics~\cite{wiebe2014hamiltonian,grimsley2019adaptive} and quantum computing, such as quantum benchmarking~\cite{scott2008optimizing,levy2021classical,huang2022foundations,blume2017demonstration}, noise characterisation~\cite{harper2020efficient}, error mitigation~\cite{strikis2021learning,em}, and variational quantum algorithms~\cite{schuld2015introduction}. Furthermore, with the crucial role of quantum computing in cryptography, another such area is cryptanalysis. In these scenarios, the quantum process of interest can manifest as a quantum oracle, providing a quantum implementation of a classical function~\cite{boneh2013quantum,kaplan2016breaking,santoli2016using,chevalier2022security,grilolearning}, or as a physical device or hardware component implementing an unknown underlying unitary, which serves as a cryptographic key or fingerprint~\cite{QPUF,CRQPUF-Original,CR-QPUF-single}. The primary challenge in learning complex quantum processes lies in the resource-intensive nature of this task, rendering conventional techniques for process tomography~\cite{mohseni2008quantum} impractical, especially for near-term devices. Recent endeavours have explored diverse techniques to devise algorithms and approaches for efficiently tackling specific instances of this challenge~\cite{arb_proc,montanaro2008quantum,fanizza2022learning,scott2008optimizing,levy2021classical,huang2022foundations, chung2018sample}.

In this work, we focus on the the \emph{statistical query} access model \cite{SQLearning}, and naturally extend it to the task of learning quantum processes. In a statistical query access model (quantum or classical) the learner constructs a hypothesis not by accessing a sequence of labelled examples themselves, but instead by adaptively querying an oracle that provides an \emph{estimate} of the statistical properties of the labelled examples. In the quantum case, this estimate is in fact the estimated expectation values of some observable. This extension to the quantum world is quite natural, as this is often the information extracted from a quantum system, by measuring it many times and estimating the expectation value of an observable, which corresponds to a physical quantity. This natural correspondence to the physics of the quantum experiment and the learning tasks designed based on them marks our main motivation for the choice of this model. In our access model, information about a process is accessed by querying it with an input state and observable, and receiving an estimate of the output expectation value. The feasibility of this model in practice, as compared to quantum PAC learning, makes it a good candidate for studying learning algorithms and their limitations in the near-term. Aside from being physically well-motivated, quantum statistical queries have also found applications in the classical verification of quantum learning~\cite{caro2023classical} and quantum error mitigation~\cite{em,arunachalam2023role}.

There are several noteworthy points regarding our proposed access model. In comparison with conventional quantum process tomography~\cite{mohseni2008quantum}, our access model is significantly weaker. In particular, we restrict the learner to measurement statistics on a single copy of the output state, and the learner cannot obtain the results of any entangled measurements with ancillary qubits. In many practical scenarios, due to the limitations of quantum memories, only the statistical results of quantum measurements are relevant and accessible in the near term. This aspect is precisely the focus of our model. Furthermore, our model gains importance when considering scenarios in which direct access to the process is not provided. This becomes especially crucial for cryptanalysis purposes, where certain attack models may not grant direct access to the process, but statistical data can still be accessed by adversaries. 

 In this access model, we study two fundamental problems in quantum learning theory. First, we look at the task of predicting properties of quantum processes, where given quantum states from a certain distribution and a list of observables, one aims to predict output expectation values after evolution under an unknown quantum process. For this task, we present an efficient learning algorithm and provide a matching lower bound (up to a logarithmic factor) in terms of the number of observables to be predicted. Next, we look at the task of learning unitaries with respect to the diamond distance, for which we provide average-case query complexity lower bounds. Furthermore, we explore the practical applications of our learning algorithm in this access model, by applying it to attack a cryptographic protocol. This result sheds light on identifying the conditions under which a class of quantum physical unclonable functions may be vulnerable, partially addressing an ongoing open question in quantum hardware security.

\subsection{Our Contributions}
In this section, we summarize our main contributions.

    \paragraph{Quantum Statistical Queries to Quantum Processes (QPSQs):} We define our access model through the quantum statistical query oracle \textsf{QPStat} for a quantum process $\mathcal{E}$, that produces an estimate of the expectation value $\tr(O \mc{E}(\rho))$. We discuss how our oracle generalises the previously defined \textsf{QStat} oracles for states. We show that in many situations, efficient quantum statistical query algorithms for learning boolean functions admit equivalent, efficient algorithms in our model. We present our access model in \Cref{sec:QPSQ-model}.

    \paragraph{Predicting properties of quantum processes:} We present an efficient algorithm that can predict properties of quantum processes from QPSQs. Our algorithm is an adaptation of the classical shadow~\cite{classical_shadow_tomography} algorithm of \cite{arb_proc} to our access model. In our weaker access model, we obtain a similar query complexity as \cite{arb_proc} with a linear overhead in the number of observables to be predicted. We show that this overhead is unavoidable, by presenting a nearly matching lower bound with respect to the number of observables. Our lower bound builds upon the technique used by \cite{arunachalam2023role} to show a similar bound for shadow tomography of quantum states from quantum statistical queries. We further demonstrate the performance of our proposed algorithm through numerical simulations. We present these results in \Cref{sec:spt}.

   \paragraph{Hardness of diamond-distance learning:} We provide an exponential lower bound on the query complexity for learning exact and approximate unitary $2$-designs from \textsf{QPStat} queries with respect to diamond distance. We also show a doubly exponential lower bound on the hardness of learning unitaries over the Haar-measure from \textsf{QPStat} queries with respect to this distance. We start by proving a general lower bound for learning any class of unitaries from QPSQs, which we obtain through a reduction from a many-vs-one distinguishing task. Our techniques are inspired by standard techniques for proving statistical query lower bounds \cite{feldman2017general}, which have been widely used in the quantum statistical query setting as well \cite{nietner2023unifying, nietner2023average, arunachalam2020quantum}. We present our lower bounds for diamond-distance learning in \Cref{sec:diamond}.
    
    \paragraph{Application to cryptanalysis:} We apply our results to cryptanalysis by studying a primitive from quantum hardware security, namely \emph{Classical Readout of Quantum Physically Unclonable Functions} (CR-QPUFs) \cite{CRQPUF-Original,CR-QPUF-single}. The security of this primitive relies on the assumed hardness of predicting statistical properties of an underlying quantum process. However, the existence of a secure realization of CR-QPUFs has not yet been shown, and remains an important open question. We partially address this problem by applying our learning algorithm for an attack against authentication protocols based on CR-QPUFs under practical physical assumptions. Our attack inherits the quasipolynomial complexity of our learning algorithm, preventing us from formally breaking the security of CR-QPUFs. However, our results present an interesting connection between learning theory and cryptography, as any new polynomial-time QPSQ algorithm for predicting expectation values would imply an attack against such protocols. 
    
\subsection{Related work}
\paragraph{Quantum Statistical Queries:}
 Quantum statistical queries (QSQs) were introduced by \cite{arunachalam2020quantum} as a quantum generalization of the statistical query (SQ) access model \cite{SQLearning}. \cite{arunachalam2020quantum} also showed efficient QSQ learning algorithms for various classes of Boolean functions which are provably hard to learn from SQs. QSQs for learning quantum states have been considered in \cite{arunachalam2023role, nietner2023unifying} and those for learning output distributions of quantum circuits have been considered in \cite{hinsche2023one,nietner2023average}. In concurrent work, \cite{angrisani2023learning} consider learning unitaries from QSQs to their Choi states, and present learning algorithms for quantum Boolean functions and quantum $k$-juntas.

\paragraph{Statistical Query Lower Bounds:}
Recent works in QSQ learning \cite{arunachalam2023role, nietner2023average, nietner2023unifying} have obtained lower bounds by adapting the classical SQ lower bound techniques of \cite{feldman2017general}, involving a reduction from a many-vs-one distinguishing task. In \cite{nietner2023average}, the authors showed lower bounds for learning the Born distributions of random circuits at varying depth regimes. \cite{nietner2023unifying} provided a unifying framework for a large class of learning models (including many 
 statistical query oracles), presented lower bounds within this framework, and applied them to the task of learning output states of quantum circuits. Compared to learning a process within diamond distance, the tasks of learning output distributions and output states are easier. However, the access model we consider is also stronger than their statistical query models. As such, our bounds cannot be directly compared to those of~\cite{nietner2023average, nietner2023unifying}. 
 
In the setting where learning algorithms are restricted to single-copy measurements,~\cite{chen2022exponential} showed an exponential sample complexity lower bound for distinguishing between Haar-random unitaries and the depolarizing channel. As QPSQ access is weaker than single-copy measurements, this immediately implies an exponential lower bound in our setting. We improve upon this by showing a doubly exponential lower bound in \Cref{thm:hardness-haar-measure}.
 
\paragraph{Predicting Properties of Quantum Processes:}

\cite{levy2021classical,kunjummen2023shadow, arb_proc, caro2022learning} have studied shadow tomography of quantum processes, i.e. predicting expectation values of observables for quantum states after evolution under an unknown quantum process. \cite{levy2021classical,kunjummen2023shadow,caro2022learning} focus on this task in the worst-case over input states, i.e., they aim to make predictions that are accurate for every input quantum state. To this end, \cite{levy2021classical,kunjummen2023shadow} use classical shadows \cite{classical_shadow_tomography} to make predictions. \cite{caro2022learning} learn elements of the Pauli transfer matrix of a channel, allowing them to efficiently make predictions for states and observables with a sparse Pauli spectrum decomposition. On the other hand, similar to our work, \cite{arb_proc} focus on an average-case version of this task, only requiring the error to be low on average over states sampled randomly. \cite{arb_proc} also make use of classical shadows \cite{classical_shadow_tomography} for this task. 

\paragraph{Concurrent work:} In independent and concurrent work, \cite{nadimpalli2024pauli} considered the similar access model of \emph{measurement queries}. Among other results, they presented an efficient algorithm in this access model for learning a class of quantum channels formed by $\mathsf{QAC}^0$ circuits.

\subsection{Organization of the paper}
We provide the necessary background and notation in Section~\ref{sec:prelim}. In Section~\ref{sec:QPSQ-model}, we define our access model and discuss its relation to other models. In Section~\ref{sec:spt}, we provide our upper and lower bound for predicting properties of an unknown quantum process. In Section~\ref{sec:diamond}, we present average-case lower bounds for learning a unitary with respect to diamond distance. Finally, we discuss applications of our algorithm for cryptanalysis in in Section~\ref{sec:crqpuf}.
\section{Preliminaries}
\label{sec:prelim}

We start by introducing the notation we use in the paper as well as the essential background. 
\subsection{Quantum Information}
We include some basic definitions of quantum computation and information in this section. For more details, we refer the reader to \cite{nielsen2002quantum}. 
We will denote the $d \times d$ identity matrix as $I_d$ and we may omit the index $d$ when the dimension is clear from the context. We use the bra-ket notation, where we denote a vector $v \in \mathbb{C}^N$ using the ket notation $\ket{v}$ and its adjoint using the bra notation $\bra{v}$. For $u,v\in \mathbb{C}^n$, we will denote by $\bra{u}\ket{v}$ the standard Hermitian inner product $u^\dag v$. A quantum (pure) state is a normalized vector $\ket{v}$, i.e. $|\bra{v}\ket{v}|=1$.
We will write $\mathcal{M}_{N,N}$ to denote the set of linear operators from $\mathbb{C}^N$ to $\mathbb{C}^N$ and we define the set of quantum states as $\mathcal{S}_N := \{\rho \in \mathcal{M}_{N,N} : \rho \succeq 0, \Tr[\rho]=1\}$.
We denote by $\mathcal{U}_N$ the set of $N$-dimensional unitary operators,
\begin{equation}
    \mathcal{U}_N:=\left\{ U\in \mathcal{M}_{N,N}: UU^\dag = U^\dag U= I\right\}.
\end{equation}
We will now introduce a useful orthonormal basis for $\mathcal{M}_{N,N}$ which is widely used in quantum information.
\begin{definition}[Pauli operators]
The set of Pauli operators is given by
    \begin{equation}
        X = \begin{pmatrix} 0 & 1 \\ 1 & 0\end{pmatrix}, Y = \begin{pmatrix} 0 & -i \\ i & 0\end{pmatrix}, Z = \begin{pmatrix} 1 & 0 \\ 0 & -1\end{pmatrix}.
    \end{equation}
\end{definition}
The set $\mathcal{P}_1=\{I,X,Y,Z\}$ forms an orthonormal basis for $\mathcal{M}_{2,2}$ with respect to the Hilbert-Schmidt inner product.
The tensor products of Pauli operators and the identity, i.e. the operators of the form $P\in \{I,X,Y,Z\}^{\otimes n}:=\mathcal{P}_n$, are usually referred as \emph{stabilizer operators} or \emph{Pauli strings} over $n$ qubits. Pauli strings over $n$ qubits form an orthonormal basis for $\mathcal{M}_{2^n,2^n}$ with respect to the Hilbert-Schmidt inner product. For a Pauli string $P \in \mc{P}_n$, we define its degree, $|P|$, as the number of indices on which it acts non-trivially, i.e., the number of non-identity Paulis.  We now look at the eigenstates of the Pauli operators, which are of special interest.
\begin{definition}[Single-qubit stabilizer states] We denote the set of Pauli eigenstates by
    \begin{equation}
        \stabone = \{ |0\rangle, |1\rangle, |+\rangle, |-\rangle, |+y\rangle, |-y\rangle \},
    \end{equation}
    where $|0\rangle$ \& $|1\rangle$ are the eigenstates of $Z$ , $|+\rangle$ \& $|-\rangle$ are the eigenstates of $X$ and $|+y\rangle$ \& $|-y\rangle$ are the eigenstates of $Y$.
\end{definition}
\begin{definition}[Clifford group] The Clifford group is the group of unitaries generated by the following 3 gates :
\begin{equation}
     H = \frac{1}{\sqrt{2}}\begin{pmatrix} 1 & 1 \\ 1 & -1\end{pmatrix}, S = \begin{pmatrix} 1 & 0 \\ 0 & i\end{pmatrix}, CNOT = \begin{pmatrix} 1 & 0 &0 & 0\\0 & 1 &0 & 0\\ 0 & 0 &0 & 1\\ 0 & 0 &1 & 0\\\end{pmatrix}.
\end{equation}
\end{definition}
We now define quantum processes.
\begin{definition}[Quantum process]
    A map $\mathcal{E} : \mathcal{M}_{N,N} \rightarrow \mathcal{M}_{N,N}$ is said to be completely positive if for any positive operator $A \in \mathcal{M}_{N^2,N^2} , (\mathcal{E} \otimes I)(A)$ is also a positive operator. $\mathcal{E}$ is said to be trace-preserving if for any input density operator $\rho, \tr(\mathcal{E}(\rho)) = \tr(\rho) = 1$. A quantum process $\mathcal{E}$ is defined as a Completely Positive Trace-Preserving (CPTP) map from one quantum state to another.   For a unitary $U$, the associated map is:
    \begin{equation}
        \mathcal{E} : \rho \rightarrow U\rho U^\dag.
    \end{equation}
\end{definition}
Next, we define the maximally depolarizing channel, which will be particularly useful.
\begin{definition}[Maximally depolarizing channel]
    The maximally depolarizing channel $\phidep$ acting on states in $S_N$ is defined as follows:
    \begin{equation}
        \phidep(\rho) = \frac{\tr(\rho)}{N}I.
    \end{equation}
\end{definition}
We now define some distances between quantum states and channels.
\begin{definition}[Trace Distance and Fidelity]
\label{def:dtr-infidelity}
    The \emph{trace distance} between two quantum states is given by
    \begin{equation}
        \dtr(\rho, \sigma) \triangleq \frac{1}{2}\norm{\rho -\sigma}_1,
    \end{equation}
    where $\|.\|_1$ is the Schatten-$1$ norm. The \emph{fidelity} between two quantum states is given by
    \begin{equation}
        F(\rho,\sigma) \triangleq \tr \left(\sqrt{\sqrt{\rho} \sigma \sqrt{\rho} }\right)^2.
    \end{equation}
    When at least one of the states is pure, $\dtr$ and $F$ are related as follows:
    \begin{equation}
    \label{ref:eq-dtr-infidel}
    1 - F(\ketbra{\psi}{\psi}, \rho) \leq \dtr(\ketbra{\psi}{\psi}, \rho).
    \end{equation}
\end{definition}
\begin{definition}[Diamond norm and diamond distance]
    \label{def:diamond}
    For a map $\mathcal{E} \in \mathcal{M}_{N,N} \rightarrow \mathcal{M}_{N,N}$, and $\mc{I} \in \mathcal{M}_{N,N} \rightarrow \mc{M}_{N,N} $ the identity superoperator, we define the diamond norm $\norm{\text{ . }}_\diamond$
    \begin{equation}
        \norm{\mathcal{E}}_\diamond = \underset{\rho \in \mathcal{S}_{N^2}}{\mathrm{max}} \norm{(\mathcal{E} \otimes \mc{I})(\rho)}_1,
    \end{equation}
    where $\|.\|_1$ denotes the Schatten $1$-norm. We then define the diamond distance $d_\diamond$ :
    \begin{equation}
        d_\diamond(\mathcal{E}_1, \mathcal{E}_2) = \frac{1}{2}\norm{\mathcal{E}_1 - \mathcal{E}_2}_\diamond.
    \end{equation}
\end{definition}
\begin{definition}[POVM]
    A \emph{Positive Operator-Valued Measure} (POVM) is a quantum measurement described by a collection of positive operators $\{E_m\}_m$, such that $\sum_m E_m = I$ and the probability of obtaining measurement outcome $m$ on a state $\rho$ is given by $p(m) = \tr(E_m\rho)$.
\end{definition}

\subsection{Unitary Haar Measure and Designs}
We now define the unitary Haar measure $\mu_H$, which can be thought as the uniform probability distribution over all quantum states or over all unitary operators in the Hilbert space of dimension $N$. For a comprehensive introduction to the Haar measure and its properties, we refer to \cite{mele2023introduction}.
\begin{definition}[Haar measure] 
The Haar measure on the unitary group $U(N)$ is the unique probability measure $\mu_H$
that is both left and right invariant over the set $\mathcal{U}_N$, i.e., for all integrable functions $f$ and for all $V \in \mathcal{U}_N$, we
have:
\begin{equation}
\int_{U(N)} f(U) d\mu_H(U) = \int_{U(N)} f(UV) d\mu_H(U) = \int_{U(N)} f(VU) d\mu_H(U).
\end{equation}
Given a state $\ket{\phi}\in\mathbb{C}^N$, we denote the $k$-th moment of a Haar random state as
\begin{equation}
\mathbf{E}_{\ket{\psi}\sim\mu_S}\left[\ket{\psi}\bra{\psi}^{\otimes k}\right]:= \mathbf{E}_{U\sim\mu_H}\left[U^{\otimes k}\ket{\phi}\bra{\phi}^{\otimes k}U^{\dag\otimes k}\right].
\end{equation}
Note that the right invariance of the Haar measure implies that the definition of $\mathbf{E}_{\ket{\psi}\sim\mu_S}\left[\ket{\psi}\bra{\psi}^{\otimes k}\right]$ does not depend on the choice of $\ket{\phi}$.
\end{definition}

A unitary $t$-design is a measure over unitaries, over which the $t$-th order moments match those of the Haar-measure exactly. In case the moments are only approximately equal, we call this an approximate unitary $t$-design. While there are many notions of approximation, here we restrict our attention to additive approximations as defined in \cite{nietner2023unifying}.

\begin{definition}[Exact and Approximate Unitary $t$-Designs]
\label{def:designs}
    The $t$-th moment superoperator with respect to a distribution $\nu$ over $\mathcal{U}_N$ is defined as
    \begin{equation}
        \mathcal{M}_\nu^{(t)}(A) = \underset{U \sim \nu}{\mathbf{E}} [U^{\otimes t} A (U^\dag)^{\otimes t}] =  \int U^{\otimes t} A (U^\dag)^{\otimes t} d\nu(U).
    \end{equation}
    $\nu$ is said to be an exact unitary $t$-design if and only if
    \begin{equation}
        \mathcal{M}_\nu^{(t)}(A) =  \mathcal{M}_{\mu_H}^{(t)}(A).
    \end{equation}
    Similarly, $\nu$ is said to be an additive $\delta$-approximate unitary $t$-design if and only if
    \begin{equation}
        \norm{\mathcal{M}_\nu^{(t)}(A) - \mathcal{M}_{\mu_H}^{(t)}(A)}_\diamond \leq \delta.
    \end{equation}
    We denote an exact unitary $t$-design by $\mu_H^{(t)}$ and an additive $\delta$-approximate unitary $t$-design by $\mu_H^{(t,\delta)}$.
\end{definition}

We will be particularly interested in the first and second-order moments over the unitary Haar measure. We start by defining the identity and flip permutation operators which will be useful for this purpose.
\begin{definition}[Identity and Flip permutation operators, Definition 12 from \cite{mele2023introduction}]
    The identity operator $\mathbb{I}$ and the flip operator $\mathbb{F}$ act on pure states $\ket{\psi}, \ket{\phi} \in \mathbb{C}^N$ as follows:
    \begin{equation}
        \mathbb{I}(\ket{\psi} \otimes \ket{\phi}) = \ket{\psi} \otimes \ket{\phi}.
    \end{equation}
    \begin{equation}
        \mathbb{F}(\ket{\psi} \otimes \ket{\phi}) = \ket{\phi} \otimes \ket{\psi}.
    \end{equation}
    We will use the following property of the flip operator. For all operators $A,B \in \mathcal{M}_{N,N}:$
    \begin{equation}
        \tr(\mathbb{F}(A \otimes B)) = \tr(AB).
    \end{equation}
\end{definition}
\begin{lemma}[First and second-order moments over the Haar-measure, Corollary 13 from \cite{mele2023introduction}]
\label{lem:haar-moments}
    We have, for $O \in \mathcal{M}_{N,N}$,
    \begin{equation}
        \mathcal{M}_{\mu_H}^{(1)}(O) = \frac{\tr(O)}{d}I .
    \end{equation}
    For $O \in \mathcal{M}_{N^2,N^2}$,
    \begin{equation}
        \mathcal{M}_{\mu_H}^{(2)}(O) = c_{\mathbb{I},O}\mathbb{I} + c_{\mathbb{F},O}\mathbb{F},
    \end{equation}
    where, 
    \begin{equation}
        c_{\mathbb{I},O} = \frac{\tr(O) - N^{-1}\tr(\mathbb{F}O)}{N^2-1} \text{ and } c_{\mathbb{F},O} = \frac{\tr(\mathbb{F}O) - N^{-1}\tr(O)}{N^2-1}.
    \end{equation}
\end{lemma}
We draw attention to the fact that $\mathcal{M}_{\mu_H}^{(1)}$ is the maximally depolarizing channel $\phidep$.

\subsection{Classical Shadow Tomography}
Classical shadow tomography is the technique of using randomized measurements to learn many properties of quantum states \cite{random_measurement,classical_shadow_tomography, huang2022provably}. It has recently been shown that classical shadow tomography can be used to predict the outcomes of arbitrary quantum processes \cite{arb_proc}. In this section, we include relevant results on classical shadow tomography that will be used in the rest of the work.
\begin{definition}[Randomized Pauli Measurement]
   Given $n>0$. A randomized Pauli measurement on an $n$-qubit state is given by a $6^n$-outcome POVM
    \begin{equation}
        \mathcal{F}^{Pauli} \triangleq \left\{ \frac{1}{3^n} \bigotimes_{i=1}^n |s_i\rangle \langle s_i| \right\}_{s_1,...,s_n \in \stabone},
    \end{equation}
    which corresponds to measuring every qubit under a random Pauli basis $(X, Y , Z$). The outcome of  $\mathcal{F}^{Pauli}$ is an $n$-qubit state $|\psi\rangle = \bigotimes_{i=1}^n |s_i\rangle ,$ where $|s_i \rangle \in \stabone$ is a single-qubit stabilizer state.
\end{definition}
 Next, we define classical shadows based on randomized Pauli measurements. Other measurements can also be used to define classical shadows.
\begin{definition}[Classical shadow of a quantum state \cite{classical_shadow_tomography}]
    Given $n,N > 0$. Consider an $n$-qubit state $\rho$. A size $N$ classical shadow of $S_N(\rho)$ of quantum state $\rho$ is a random set given by 
    \begin{equation}
        S_N(\rho) \triangleq \left\{ |\psi_l \rangle \right\}_{l = 1}^N,
    \end{equation}
    where $|\psi_l\rangle = \bigotimes_{i=1}^n |s_{l,i}\rangle$ is the outcome of the $l$-th randomized Pauli measurement on a single copy of $\rho$.
\end{definition}
\begin{definition}[Classical Shadow Approximation of a quantum state \cite{classical_shadow_tomography}]
    Given the classical shadow $S_N(\rho)$ of an $n$-qubit state $\rho$. We can approximate $\rho$ via 
    \begin{equation}
        \sigma_N(\rho) = \frac{1}{N} \sum_{l = 1}^N \bigotimes_{i = 1}^n (3|s_{l,i}\rangle\langle s_{l,i}|-I),
    \end{equation}
\end{definition}
\begin{definition}[Classical shadow of a quantum process \cite{arb_proc}]
\label{def:classical-shadow-process}
    Given an $n$-qubit CPTP map $\mathcal{E}$. A size-$N$ classical shadow $S_N(\mathcal{E})$ of the quantum process $\mathcal{E}$ is a random set given by 
    \begin{equation}
        S_N(\mathcal{E}) \triangleq \left \{ |\psi_l^{(in)}\rangle, |\psi_l^{(out)}\rangle \right \}_{l=1}^N,
    \end{equation}
    where $|\psi_l^{(in)}\rangle = \bigotimes_{i=1}^n |s_{l,i}^{(in)}\rangle$ is a random input state with $|s_{l,i}^{(in)}\rangle \in stab_1$ sampled uniformly at random, and $|\psi_l^{(out)}\rangle = \bigotimes_{i=1}^n |s_{l,i}^{(out)}\rangle$ is the outcome of performing a random Pauli measurement on $\mathcal{E}(|\psi_l^{(in)}\rangle \langle \psi_l^{(in)}|)$.
\end{definition}
 The authors in \cite{arb_proc} recently proposed a machine learning algorithm that is able to learn the average output behaviour of any quantum process, under some restrictions. In the learning phase, the algorithm works with the classical shadow of a generic quantum process $\mathcal{E}$ and a set of observables $\{O_i\}$. In the prediction phase, the algorithm receives an input quantum state $\rho$ sampled from the target distribution $D$, and aims to predict $\tr(O_i\mathcal{E}(\rho))$ for all observables in the set. The algorithm comes with a rigorous performance guarantee on the average prediction error over $D$, achieved with efficient time and sample complexity with respect to the number of qubits and error parameters. While the guarantee holds for any quantum process, there are certain restrictions on the observables $\{O_i\}$ and the distribution $D$. We partially state their results in \Cref{lem:avg-spt-from-heis-obs} and refer to~\cite{arb_proc} for further details.
 
\subsection{Computational Learning Theory}
Computational learning theory studies what it means to \textit{learn a function}. One of the most successful formal learning frameworks is undoubtedly the model of \emph{Probably Approximately Correct} (PAC) learning, which was introduced in \cite{valiant1984theory}. In this model, we consider a class of target Boolean functions $\mathcal{C} \subseteq \{f | f:\{0,1\}^n \rightarrow \{0,1\}\}$,  usually called the \textit{concept class}. For an arbitrary concept $c \in \mathcal{C}$, a PAC learner receives samples of the form  $\{x,c(x)\}$, where, in general, $x$ is sampled from an unknown probability distribution $D: \{0,1\}^n \rightarrow [0,1]$. In the setting of noisy PAC learning, the bit $c(x)$ of each sample may independently be incorrect with some probability. The learner aims to output, with high probability, a hypothesis function $h$ with low error on average over the distribution $D$.

Another widely studied access model is that of learning from Statistical Queries (SQs) \cite{SQLearning}. Here,  a learner is more restricted in the way it can interact with the data. Rather than learning from individual samples, the algorithm learns using the statistical properties of the data, making it more robust to noise. In particular, an SQ learner receives as input estimates of the expectation values of some chosen functions within specified error tolerance. 

Quantum generalizations of both PAC and SQ learning have already been introduced and studied widely. The Quantum PAC (QPAC) model was introduced in \cite{QPAC}, where the learner has access to a quantum computer and receives \textit{quantum example states} as input. The quantum example state for a concept $c$ over $n$ input bits with the target distribution $D$ is the $n+1$-qubit state $|\psi_c\rangle = \sum_{x}\sqrt{D(x)}|x,c(x)\rangle$. It has been shown in~\cite{arunachalam2018optimal} that the sample complexity of quantum and classical PAC learning under unknown distributions is the same. However, over a fixed uniform distribution, learning with quantum queries can provide exponential advantage over the classical learner~\cite{atici2005improved,grilolearning}. Efficient learners from quantum queries under product distributions have also been shown \cite{kanade2019a, Caro_2020}.

 A quantum analogue of statistical queries was introduced in \cite{arunachalam2020quantum}. Here, the statistical query returns an approximation of the expectation value for an input measurement observable on quantum examples of the concept class to be learned. We include the quantum statistical query oracle defined in \cite{arunachalam2020quantum} below:
\begin{definition}[\textsf{QStat}, from \cite{arunachalam2020quantum}]
\label{def:QSQL}
    Let $\mathcal{C} \subseteq \{c : \{0,1\}^n \rightarrow \{0,1\} \}$ be a concept class and $D : \{0,1\}^n \rightarrow [0,1] $ be a distribution over $n$-bit strings. A quantum statistical query oracle $\mathsf{QStat}(O,\tau)$ for some $c^* \in \mathcal{C}$ receives as inputs $O, \tau$, where $\tau \geq 0$ and $O \in (\mathbb{C}^2)^{\otimes n+1} \times (\mathbb{C}^2)^{\otimes n+1}, ||O||_\infty \leq 1$, and returns a number $\alpha$ satisfying
    \begin{equation*}
        |\alpha - \langle \psi_{c^*}|O|\psi_{c^*}\rangle| \leq \tau,
    \end{equation*}
    where $|\psi_{c^*}\rangle = \sum_{x \in \{0,1\}^n} \sqrt{D(x)}|x,c^*(x)\rangle$.
\end{definition}
Note that in the QSQ model of~\cite{arunachalam2020quantum}, while the learner can obtain an estimate of any measurement on the quantum examples, it is restricted only to classical computation.
Interestingly, several concept classes such as parities, juntas, and DNF formulae are efficiently learnable in the QSQ model, whereas the classical statistical query model necessitates an exponentially larger number of queries. Additionally, the authors of \cite{arunachalam2023role} have established an exponential separation between QSQ learning and learning with quantum examples in the presence of classification noise.

One can also define quantum statistical queries for the task of learning states. In this case, the quantum example in~\Cref{def:QSQL} would be replaced by the unknown quantum state to be learned.

Quantum statistical queries have also found practical applications in classical verification of quantum learning, as detailed in \cite{caro2023classical}. Furthermore, they have been employed in the analysis of quantum error mitigation models \cite{em,arunachalam2023role} and quantum neural networks \cite{Du_2021}. Alternative variations of quantum statistical queries have also been explored in \cite{hinsche2023one, gollakota2022hardness, nietner2023average}.
\section{Quantum Statistical Queries to Quantum Processes}\label{sec:QPSQ-model}

In this section, we propose a framework and definition for learning quantum processes through quantum statistical queries and discuss its importance and relevance to different problems. Previously studied quantum statistical queries (Definition \ref{def:QSQL}) have considered queries to quantum examples associated with some classical function \cite{arunachalam2020quantum}. While this model is the first generalisation of statistical query learning into the quantum setting and this type of query is useful for the problem of learning certain classes of classical functions, they do not encompass the quantum processes, and as such a new framework is needed for studying the learnability of quantum processes through quantum statistical queries. Learning quantum processes from (often limited amount of) data is a crucial problem in physics and many areas of quantum computing such as error characterisation and error mitigation~\cite{strikis2021learning,kim2020quantum,mohseni2008quantum,harper2020efficient,huang2021information}. In many realistic and near-term scenarios, the only accessible data of the quantum process is through measured outcomes of such quantum channels, which are in fact nothing but statistical queries to such quantum processes. Hence studying the quantum process learnability via statistical queries is well motivated practically from the nature of quantum experiments and measurements. In what follows, we define our access model and then clarify its relationship to the previous definition of QSQs.
\begin{definition}[Statistical Query to a Quantum Process (QPSQ)]
\label{def:SQ}
Let $\mathcal{E}: \mathbb{C}^d \rightarrow \mathbb{C}^d$ be a quantum process acting on a $d$-dimensional Hilbert space. A QPSQ learning algorithm has access to a quantum statistical query oracle \textsf{QPStat} of the process $\mathcal{E}$, which receives as input an observable $O \in \mathbb{C}^d \times \mathbb{C}^d$ satisfying $\norm{O}_\infty \leq 1$, a quantum state $\rho \in \mathcal{S}_d$, and a tolerance parameter $\tau \geq 0$, and outputs a number $\alpha$ satisfying
\begin{equation}
     \\|\alpha - \tr(O\mathcal{E}(\rho))| \leq \tau.
\end{equation}
We denote the query as $\alpha \leftarrow \mathsf{QPStat}_{\mathcal{E}}(\rho,O, \tau)$\footnote{When referring to unitaries, we will sometimes abuse the notation and use $\mathsf{QPStat}_{U}$ to refer to the oracle associated with the unitary channel $\rho \rightarrow U \rho U^\dag$.}. The output $\alpha$ acts as an estimate of the expectation value of $O$ on the state $\rho$ after evolution under $\mathcal{E}$ within absolute error $\tau$.  
\end{definition}
 Our definition of QPSQs is justified in the setting where a learner has black-box access to a quantum process, with the ability to query the process with any quantum state. 
We note that our definition does not specify how the input quantum state $\rho$ is provided to the oracle. The algorithm can provide multiple copies of the quantum state or can send the classical description of the quantum state to the oracle where they can be locally prepared, depending on the scenario and application. Obtaining the output of these queries is then achieved by the most natural operation, which is estimating the desired observable after evolution under the process. Again, since most physical properties of a quantum system are extracted through such interactions with its associated quantum channel, the application of this model in physics is straightforward. However, we will show that this model can be used in various scenarios including quantum cryptanalysis (see \Cref{sec:crqpuf}) and learning quantumly-encoded classical functions. 

We now discuss what it means for a QPSQ learning algorithm to be \textit{efficient}. Any such learning algorithm aiming to learn some property of the process $\mathcal{E}$ must be a quantum polynomial time (\textit{QPT}) algorithm, making at most polynomially many queries to the oracle \textsf{QPStat}. As the learner itself must provide the input states or their classical description to \textsf{QPStat}, it must be possible to efficiently prepare the required copies of these states. This is an important point for the physical justification of this model. If the learner was able to query \textsf{QPStat} with arbitrary quantum states that require exponential quantum computation for preparation, the learning model would no longer be physically justified. On a similar note, the observables that an efficient learner provides to the oracle must also be efficiently measurable. Then, in the following definition for an \textit{efficient} QPSQ learning algorithm, we only discuss the efficiency of the algorithm, not its correctness, allowing various notions of correctness depending on the desired property of $\mathcal{E}$ to be learned up for consideration.
\begin{definition}[Efficient QPSQ learner]
\label{def:QPSQ-learner}
A QPSQ learning algorithm is called an \emph{efficient} QPSQ learner if it makes at most $\mathsf{poly}(\log(d))$ queries with tolerance at least $1/\mathsf{poly}(\log(d))$ to the \textsf{QPStat$_\mathcal{E}$} oracle with states preparable in $\mathsf{poly}(\log(d))$ time, and observables measurable up to precision $\tau$ in $\mathsf{poly}(\log(d), 1/\tau)$ time, and runs in $\mathsf{poly}(\log(d))$ computational time.
\end{definition}
After formally defining our QPSQ model, we now talk about the relationship between our proposed model and other SQ models. It is easy to see that the \textsf{QPStat} oracle in Definition~\ref{def:SQ} generalizes the \textsf{QStat} oracle from Definition~\ref{def:QSQL}. We start by considering the unitary $U_c$ associated with a Boolean function $c : \{0,1\}^n \rightarrow \{0,1\}$, $U_c |x,y\rangle = |x, y \oplus c(x)\rangle, \text{  } \forall x \in \{0,1\}^n, y \in \{0,1\}$, 
and the quantum process $\mathcal{E}_c : \rho \rightarrow U_c\rho U_c^\dag $. Let $|\psi_D\rangle = \sum_{x \in \{0,1\}^n} \sqrt{D(x)}|x,0\rangle$ be a superposition state associated with a distribution $D$ over $\{0,1\}^n$. Clearly, $\mathcal{E}_c(| \psi_D \rangle \langle \psi_D|) = | \psi_c \rangle \langle \psi_c|$. Thus, for any observable $O$, we can see that $\Tr(O\mathcal{E}_c(| \psi_D \rangle \langle \psi_D|)) = \langle \psi_c|O|\psi_c\rangle$, giving us the equivalence
\begin{equation}
    \textsf{QPStat}_{\mathcal{E}_c}(| \psi_D \rangle \langle \psi_D|, O, \tau) \equiv\textsf{QStat}_{\ket{\psi_c}}( O,\tau).
\end{equation}
Along with the definition of \textsf{QStat}, Arunachalam et. al. \cite{arunachalam2020quantum} presented algorithms for learning various concept classes in their QSQ model. The generalization of \textsf{QStat} by \textsf{QPStat} implies that these algorithms also hold in our QPSQ learning model. Given \textsf{QPStat} oracle access to the process $\mathcal{E}_c$ associated with the target concept, and sufficient copies of the state $|\psi_D\rangle$, the required output from \textsf{QStat} queries can be obtained using \textsf{QPStat} queries, and the remainder of the learning algorithms proceed identically. 
\begin{remark}
    Any concept class efficiently learnable in the QSQ model under a \textit{known distribution} $D$ can be learned efficiently given QPSQ access to the unitary encoding of the target function instead if $|\psi_D\rangle$ \textit{can be prepared efficiently}. Thus, by extending the algorithms from~\cite{arunachalam2020quantum}, one can show that \textit{there exist efficient QPSQ algorithms for learning parities, juntas and DNFs under the uniform distribution}. The generalization also holds in the \textit{unknown distribution} setting assuming access to copies of $|\psi_D\rangle$.
\end{remark}
\section{Predicting Properties of Quantum Processes}
\label{sec:spt}
In this section, we consider the problem of predicting properties of quantum processes, i.e., shadow process tomography, previously studied by \cite{kunjummen2023shadow,levy2021classical,caro2022learning}. We define the problem as follows.
\begin{problem}[Shadow Tomography of a Quantum Process]
    \label{prb:spt-worst}
    Let $0 < \epsilon, \delta < 1$ and $ M \geq 1$. Given access to an unknown quantum process $\mc{E}$, and a list of quantum states and observables $\{(\rho_i, O_i)\}_{i \in [M]}$, produce predictions $\{h(\rho_i,O_i)\}_{i \in [M]}$ such that, with probability at least $1-\delta$,
    \begin{equation}
        |h(\rho_i,O_i) - \tr(O_i \mc{E}(\rho_i))| \leq \epsilon, \forall i \in [M].
    \end{equation}
\end{problem}
We will focus primarily on the average-case relaxation of this problem \cite{arb_proc}, where the error only needs to be low on average over random input states, rather than on all input states.   
\begin{problem}[Average-case Shadow Process Tomography]
\label{prb:spt-avg}
    Let $0 < \epsilon, \delta < 1$ and $ M \geq 1$. Given access to an unknown quantum process $\mc{E}$, and a list of observables $\{O_i\}_{i \in [M]}$ as well as quantum states drawn randomly from some distribution $\mc{D}$, produce predictions $h(\rho, O_i)$ such that, with probability at least $1-\delta$,
    \begin{equation}
        \mathbf{E}_{\rho \sim \mc{D}} \left[|h(\rho, O_i)-\tr(O\mc{E}(\rho))|^2\right] \leq \epsilon^2, \forall i \in [M].
    \end{equation}
\end{problem}
 We state our main result on Problem \ref{prb:spt-avg} from access to $\mathsf{QPStat}$ queries in the following theorem.
\begin{theorem}[Average-case Shadow Process Tomography from QPSQs]
\label{thm:spt-avg}
    There exists an algorithm for solving Problem \ref{prb:spt-avg} for $M$ observables $\{O_i\}_{i \in [M]}$ and a distribution $\mc{D}$, where each $O_i$ is an $n$-qubit observable with $\|O_i\|_\infty \leq 1$, given as a sum of few-body $(\leq \kappa  = \mathcal{O}(1))$ observables, where each qubit is acted on by $\mathcal{O}(1)$ of the few-body observables, and $\mc{D}$ is a distribution over quantum states that is invariant under single-qubit Clifford operations, using
    \begin{equation}
        N = M \log(Mn/\delta) 2^{\mc{O}(\log (n) \log(1/\epsilon))}
    \end{equation}
    queries of tolerance $1/2^{\mc{O}(\log (n) \log(1/\epsilon))}$ to $\mathsf{QPStat}_\mc{E}$ and computational time $N \cdot 2^{\mc{O}(\log (n) \log(1/\epsilon))}$.
    Moreover, any algorithm solving Problem \ref{prb:spt-avg} with high probability for $M$ Pauli observables and any distribution over target quantum states must make
    \begin{equation}
        \Omega \left( \frac{M \tau^2}{\epsilon^2}\right)
    \end{equation}
    $\mathsf{QPStat}$ queries of tolerance $\tau$ to the unknown channel. 
\end{theorem}
We prove our upper bound by adapting the learning algorithm from \cite{arb_proc} to the QPSQ setting, resulting in a factor-$M$ overhead in the query complexity compared to their result. Our lower bound complements this by showing that this overhead is necessary - any QPSQ algorithm for average-case shadow tomography of any quantum process must make $\Omega(M)$ queries. Further, as the lower bound holds for any $M$ Paulis and any distribution over states, it also applies to the kinds of local observables and distributions considered in the upper bound. Thus, we see that with respect to the number of observables, the upper bound is tight up to a logarithmic factor. Moreover, for $\epsilon = \Theta(1)$, the query and time complexities of the upper bound scale polynomially with $n$.  

The proof of our lower bound uses a similar construction to the lower bound for shadow tomography of quantum states from QSQs shown in \cite{arunachalam2023role}. For our upper bound, we adapt the learning algorithm of \cite{arb_proc} to the QPSQ setting. We present the proof of our lower bound in \Cref{sec:spt-avg-lower} and that of our upper bound in \Cref{sec:spt-avg-upper}. We support the theoretical guarantees of our learning algorithm with numerical simulations in \Cref{sec:simulations}.
\subsection{Proof of Lower Bound}\label{sec:spt-lower}
\label{sec:spt-avg-lower}
In this section, we prove the lower bound of \Cref{thm:spt-avg} on average-case shadow process tomography (Problem \ref{prb:spt-avg}) from $\mathsf{QPStat}$ queries. We start by presenting a lower bound for Problem \ref{prb:spt-worst} when the observables are Paulis. First, we recall a result from \cite{arunachalam2023role} on Pauli Shadow Tomography of quantum states from QSQs.
\begin{lemma}[Adapted from Theorem 27 of \cite{arunachalam2023role}]
\label{lem:qsq-st-worst}
    Let $C \subseteq \mc{P}_n,$ with $ |C| = M$, be a set of $n$-qubit Pauli operators. Let $S$  be the associated class of states
    \begin{equation}
        \mc{S}_{\epsilon} = \left\{ \rho_{\epsilon,P} = \frac{\mathbb{I}+3\epsilon P}{2^n}, P \in \mc{C} \right\}.
    \end{equation}
    Then, any algorithm that distinguishes between $\rho \in \mc{S}_\epsilon$ and $\rho = \frac{\mathbb{I}}{2^n}$ must make $\Omega(M\tau^2/\epsilon^2)$ queries of tolerance $\tau$ to $\mathsf{QStat}_\rho$.
\end{lemma}
While \cite{arunachalam2023role} showed the above result for the specific case of $\mc{C} = \mc{P}_n$, their proof can easily be adapted to show this for any subset of $\mc{P}_n$. Further, \cite{arunachalam2023role} noted that this distinguishing task reduces to shadow tomography of the set of Pauli operators, resulting in a QSQ lower bound for the latter. We construct a similar distinguishing task for quantum channels, giving us a lower bound for Problem \ref{prb:spt-worst}.
\begin{prop}[Worst-Case Hardness of Shadow Process Tomography from QPSQs]
\label{prop:spt-worst}
    Any algorithm solving Problem \ref{prb:spt-worst} for $M$ Pauli observables and arbitrary input states within precision $\epsilon$ must make $\Omega(M\tau^2/\epsilon^2)$ queries of tolerance $\tau$ to $\mathsf{QPStat}_{\mc{E}}$.
\end{prop}
\begin{proof}
    Let $C \subseteq \pauli_n$, with $|C| = M$, be the set of $M$ Pauli observables. Consider the class of channels $\mc{C}_\epsilon$ that prepare states from $\mc{S}_{\epsilon}$ (\Cref{lem:qsq-st-worst}), i.e.,
    \begin{equation}
        \mc{C}_\epsilon = \left\{ \Phi_{\epsilon, P}: \rho_{\mathrm{in}} \rightarrow \tr(\rho_{\mathrm{in}}) \frac{\mathbb{I}+3\epsilon P}{2^n} \right\}.
    \end{equation}
    Given access to $\mathsf{QPStat}_\mc{E}$ queries of tolerance $\tau$, and promised that $\mc{E} \in \mc{C}_\epsilon$ or $\mc{E} = \phidep$, we consider the task of determining which is the case. Clearly, for any input state $\rho$, and Paulis $P,Q \in \mc{C}$,
    \begin{equation}
        \tr(P \Phi_{\epsilon,Q}(\rho)) = 3\epsilon\delta_{P,Q}.
    \end{equation}
    Meanwhile, for any Pauli $P \in \mc{C}$,
    \begin{equation}
        \tr(P \phidep (\rho)) = 0.
    \end{equation}
    Thus, any algorithm for solving Problem \ref{prb:spt-worst} for all observables in $\mc{C}$ and arbitrary input states, implies an algorithm for distinguishing between $\mc{C}_\epsilon$ and $\phidep$.
    
    Now, to lower bound the query complexity of this distinguishing task, we note that the channels in this problem simply discard the input state and prepare a fixed output state. Suppose the output state of a channel $\mc{E} \in \mc{C}_\epsilon \cup \{\phidep\}$ is $\rho_{\mathrm{out}}$. Clearly, the responses of $\mathsf{QPStat}_\mc{E}(\rho_{\mathrm{in}}, O, \tau)$ are indistinguishable from those of $\mathsf{QStat}_{\rho_{\mathrm{out}}}(O, \tau)$. Recall that the class of output states of $\mc{C}_\epsilon$ is exactly $\mc{S}_\epsilon$ and that the output of $\phidep$ is always $\frac{\mathbb{I}}{2^n}$. Thus, any algorithm for distinguishing $\mc{C}_\epsilon$ versus $\phidep$ using $\mathsf{QPStat}$ queries can be used to distinguish between $\mc{S}_\epsilon$ and $\frac{\mathbb{I}}{2^n}$ using the same number of $\mathsf{QStat}$ queries. Together with \Cref{lem:qsq-st-worst}, we obtain the desired lower bound.
\end{proof}
 While we have shown a lower bound for worst-case shadow tomography of quantum processes (Problem~\ref{prb:spt-worst}), Problem \ref{prb:spt-avg} may be easier as it only requires the error to be low on average over states, rather than on any input quantum state. We prove the lower bound by showing that the same distinguishing task from Proposition \ref{prop:spt-worst} can be solved by an algorithm for Problem \ref{prb:spt-avg} without additional queries.
\begin{proof}[Proof of \Cref{thm:spt-avg} (Lower Bound)]
   Suppose there exists an algorithm $\mc{A}$ for Problem \ref{prb:spt-avg} from $\mathsf{QPStat}$ queries. Let $C \subseteq \pauli_n$, with $|C| = M$, be the set of $M$ Pauli observables. We apply $\mc{A}$ to a channel $\mc{E} \in  \mc{C}_{\epsilon} \cup \{\phidep\}$.
   Then, from the correctness of $\mc{A}$, we obtain a hypothesis $h$, such that for the target distribution $\mc{D}$ over input states,
   \begin{equation}
       \mathbf{E}_{\rho \sim D} \left[|h(\rho, P) - \tr(P \mc{E}(\rho))|^2\right] \leq \epsilon^2, \forall P \in \mc{C}.
   \end{equation}
   For $\mc{E} \in  \mc{C}_{\epsilon} \cup \{\phidep\}$, the output state $\mc{E}(\rho)$ is independent of $\rho$. Let $\alpha_P = \tr(P\mc{E}(\rho))$ for any input $\rho$. Then, we have
   \begin{equation}
   \label{eq:avg-error-bound}
       \mathbf{E}_{\rho \sim D} \left[|h(\rho, P) - \alpha_P|^2\right] \leq \epsilon^2, \forall P \in \mc{C}.
   \end{equation}
   Now, we compute the mean of the hypothesis function, $\mu_P = \mathbf{E}_{\rho \sim \mc{D}} [h(\rho,P)]$. While this step may be computationally expensive, it does not require any additional $\mathsf{QPStat}$ queries. Our distinguisher will use $\mu_P$ as an estimate for $\alpha_P$. We bound the error as follows.
   \begin{align}
       |\mu_P - \alpha_P| &= \left|\mathbf{E}_{\rho \sim D} \left[h(\rho, P) - \alpha_P\right] \right|
       \\&\leq \sqrt{ \mathbf{E}_{\rho \sim D} \left[|h(\rho, P) - \alpha_P|^2\right] }
       \\&\leq \epsilon
   \end{align}
   where the first inequality follows from Jensen's inequality, and the second inequality follows from \Cref{eq:avg-error-bound}. Thus, using $\mc{A}$, we can estimate $\alpha_P$ within error $\epsilon$  for all $P \in \mc{C}$, without any additional $\mathsf{QPStat}$ queries. From the proof of Proposition \ref{prop:spt-worst}, we see that this is enough to distinguish between $\mc{C}_{\epsilon}$ and $\phidep$. Together with Proposition \ref{prop:spt-worst}, we now obtain the stated lower bound.
\end{proof}
\subsection{Proof of Upper Bound}
\label{sec:spt-avg-upper}
In this section, we present an algorithm achieving the upper bound of \Cref{thm:spt-avg}. The algorithm follows from a straightforward adaptation of the classical shadow algorithm for this problem by \cite{arb_proc}. The key result of \cite{arb_proc} we use is that to succeed at Problem \ref{prb:spt-avg}, it suffices to estimate (in a certain sense) the low-degree Pauli coefficients of the Heisenberg-evolved observables. We state this result in the following lemma.  
\begin{lemma}[Adapted from \cite{arb_proc}]
\label{lem:avg-spt-from-heis-obs}
    Given $\epsilon,\delta > 0$, a distribution $\mc{D}$ over quantum states that is invariant under single-qubit Clifford operations, and an $n$-qubit observable $O$ given as a sum of few-body observables with degree $\mc{O}(1)$, where each qubit is acted on by $\mathcal{O}(1)$ of the few-body observables and $\|O\|_\infty \leq 1$. Define hyperparameters $k, \Tilde{\epsilon}$, where
    \begin{equation}
        \label{eq:alg-hyperparam}
        k = \lceil log_{1.5}(2/\epsilon^2)\rceil, \quad \Tilde{\epsilon} = \Theta\left(\frac{\epsilon^2}{(2n)^k}\right).
    \end{equation}
    Let $\mc{E}$ be the unknown process and consider the Heisenberg-evolved observable $\mc{E}^\dagger (O)$ with Pauli coefficients $\alpha_P(O)$, i.e.,
    \begin{equation}
        \mc{E}^\dagger (O) = \sum_{P \in \mc{P}_n} \alpha_P(O) P.
    \end{equation}
    Then, given the description of a random state $\rho \sim \mc{D}$, in order to achieve low average-case error over $\mc{D}$,
    \begin{equation}
        \mathbf{E}_{\rho \sim \mc{D}} \left[|h(\rho, O)-\tr(O\mc{E}(\rho))|^2\right] \leq \epsilon^2,
    \end{equation}
    it suffices to obtain $\Tilde{\epsilon}$-accurate estimates of $\frac{1}{3^{|P|}} \alpha_P(O)$ for all Paulis $P$ with degree $|P| \leq k$ along with $\mc{O}(k (3n)^k)$ computational time for each prediction.
\end{lemma}
We note that \cite{arb_proc} additionally consider another setting for the hyperparameters $k, \Tilde{\epsilon}$ when an additional error parameter is included. For simplicity, we only consider the hyperparameter setting stated in \Cref{eq:alg-hyperparam}, but our analysis can be extended to the other setting in \cite{arb_proc}. 

Now, from \Cref{lem:avg-spt-from-heis-obs}, we see that we need to learn the low-degree Pauli coefficients of $\mc{E}^\dagger (O)$ from queries to $\mathsf{QPStat}$. Apart from this step, the rest of the algorithm proceeds identically to that of \cite{arb_proc}. For completeness, we detail the full algorithm in \Cref{alg:avg-spt}, where we only discuss the case for a single observable. This can be extended to $M$ observables by repeating the algorithm for each observable and increasing the number of queries by readjusting the failure probability of each repetition to $\delta/M$.
In \Cref{alg:avg-spt}, we make use of $\|O\|_{\mathrm{Pauli,1}}$, which denotes the $l_1$ norm of the Pauli coefficients of $O$. Specifically, for $O = \sum_{P \in \pauli_n} a_P P$,  we have $\|O\|_{\mathrm{Pauli,1}} = \sum_{P \in \pauli_n} |a_P|$.
\begin{algorithm}
\caption{Learning to predict properties of a quantum process from QPSQs}
    \label{alg:avg-spt}
    \begin{algorithmic}
     \State $k \gets \lceil log_{1.5}(2/\epsilon^2)\rceil, \Tilde{\epsilon} \gets \Theta\left(\frac{\epsilon^2}{(2n)^k}\right)$
     \\
    \State \textbf{Gather Data}:
        \For{$l = 1$ to $N$}
            \State $|\psi^{(in)}_l\rangle \gets \bigotimes_{i = 1}^n |s^{(in)}_{l,i}\rangle, |s^{(in)}_{l,i}\rangle \in \stabone$, chosen uniformly at random 
            \State $y_l \gets \mathsf{QPStat}_{\mathcal{E}}(|\psi^{(in)}_l\rangle \langle \psi^{(in)}_l|, O,  \tau)$
        \EndFor
        \State \textbf{return} $S_N(\mathcal{E}, O) = \{|\psi^{(in)}_l\rangle, y_l\}_{l = 1}^N$
    \\
    \State \textbf{Learning}:
    \ForAll{$P \in \mathcal{P}_n, |P| \leq k$}
        \State $\hat{x}_P(O) \gets \frac{1}{N} \sum_{l = 1}^N y_l \tr(P |\psi_l^{in}\rangle \langle \psi_l^{in}|) $
        \If{ 
        $(\frac{1}{3})^{|P|} > 2\Tilde{\epsilon}$ and $ |\hat{x}_P(O)| > 2.3^{|P|/2}\sqrt{\Tilde{\epsilon}}\|O\|_{\mathrm{Pauli,1}}$
        }
        \State $\hat{\alpha}_P(O) \gets 3^{|P|} \hat{x}_P(O)$
        \Else 
        \State $\hat{\alpha}_P(O) \gets 0$
        \EndIf
    \EndFor
    \\
    \State \textbf{Prediction} for target state $\rho$:
    \State $h(\rho) \gets \sum_{P : |P| \leq k} \hat{\alpha}_P(O) \tr(P\rho)$
    \State \textbf{return} $h(\rho)$
\end{algorithmic}
\end{algorithm}
While \Cref{lem:avg-spt-from-heis-obs} requires the classical description of a state $\rho$ to make predictions on it, from \Cref{alg:avg-spt}, we see that it suffices to have estimates of $\tr(P\rho)$ for all low-degree Paulis $P \in \pauli_n, |P| \leq k$. Instead of a complete classical description, this could also be estimated using classical shadows \cite{classical_shadow_tomography} or from $\mathsf{QStat}_\rho$ queries, depending on the access available.

\begin{proof}[Proof of \Cref{thm:spt-avg} (Upper Bound)]
From \Cref{lem:avg-spt-from-heis-obs}, it suffices to estimate the Pauli coefficients $\frac{1}{3^{|P|}}\alpha_P(O)$ within precision $\Tilde{\epsilon}$. Let $\mathrm{stab}_1^{\otimes n}$ be the uniform distribution over the tensor product of $n$ single-qubit stabilizer states. Then, \cite{arb_proc} showed that for all Paulis $P \in \pauli_n$,
\begin{equation}
    \frac{1}{3^{|P|}}\alpha_P(O) = \underset{\rho \sim \mathrm{stab}_1^{\otimes n}}{\mathbf{E}} \tr(P\rho)\tr(O\mathcal{E}(\rho)).
\end{equation}
Denote $x_P(O) \triangleq  \frac{1}{3^{|P|}}\alpha_P(O)$. We construct estimators for $x_P(O)$ by performing queries to $\mathsf{QPStat}$ with random states $\rho_l \sim \mathrm{stab}_1^{\otimes n}, l \in [N]$. Let $y_l \gets \mathsf{QPStat}_{\mc{E}}(\rho_l, O, \tau)$. Then, for all low-degree Paulis $P \in \pauli_n, |P| \leq k$, we construct estimators
\begin{equation}
    \hat{x}_P(O) = \frac{1}{N} \sum_{l \in [N]} \tr(P \rho_l) y_l. 
\end{equation}
Before bounding the error between $\hat{x}_P(O)$ and $x_P(O)$, we first define the intermediate random variable $x_P^\prime(O)$
\begin{equation}
    x_P^\prime(O) \triangleq \frac{1}{N}\sum_{l=1}^N \tr(P\rho_l)\tr(O\mathcal{E}(\rho_l)).
\end{equation}
Now, we bound the error between $x_P(O)$ and $x_P^\prime(O)$. Using $N = \mc{O}\left(\frac{\log((3n)^k/\delta)}{(\Tilde{\epsilon} - \tau)^2}\right)$, from Hoeffding's inequality and assuming $\tau < \Tilde{\epsilon}$, we see that for any $P \in \pauli_n, |P| \leq k$, with probability at least $1 - \delta/(3n)^k$,
\begin{equation}
\label{eq:alg-triangle-pt1}
    |x_P(O) - x_P^\prime(O)| \leq \Tilde{\epsilon} - \tau.
\end{equation}
Using a union bound over failure probabilities, and noting that there are at most $(3n)^k$ Pauli operators of degree $k$, we see that \Cref{eq:alg-triangle-pt1} holds for all $P \in \pauli_n, |P| \leq k$ with probability at least $1-\delta$. For the following arguments, we condition on this event. Next, we bound the error between $x_P^\prime(O)$ and $\hat{x}_P(O)$.
\begin{align}
    |x_P^\prime(O) - \hat{x}_P(O)| &= \left|
    \frac{1}{N} \sum_{l \in [N]} \tr(P \rho_l) \left(\tr(O\mc{E}(\rho_l)) -y_l\right)
    \right| \\&\leq \frac{1}{N} \sum_{l \in [N]}|(\tr(O\mc{E}(\rho_l)) -y_l| \\& \label{eq:alg-triangle-pt2} \leq \tau,
\end{align}
where the first inequality uses the triangle inequality and that $\|P\|_\infty = 1$. Now, applying the triangle inequality to \Cref{eq:alg-triangle-pt1,eq:alg-triangle-pt2}, we have
\begin{equation}
    |x_P(O) -  \hat{x}_P(O)| \leq \Tilde{\epsilon},
\end{equation}
which is the desired bound on the error. The stated query complexity for $M = 1$ can be obtained by setting $\tau = \Tilde{\epsilon}/2$ and selecting $k$ and $\Tilde{\epsilon}$ as stated in \Cref{eq:alg-hyperparam}. To obtain the query complexity for general $M$, we can repeat the entire algorithm for each observable with failure probability at most $\delta/M$.

The computational complexity is dominated by estimating the Pauli coefficients, each of which requires time $\mc{O}\left(kN\right)$. The stated complexity can be obtained by noting that we estimate at most $(3n)^k$ coefficients.
 \end{proof}
In certain practical cases, the following additional assumption about the output $\alpha$ of the oracle $\mathsf{QPStat}_{\mathcal{E}}$ may hold.
\begin{assumption} \label{ass:qpstat-assumption}
    Consider any input state $\rho$, observable $O$ and tolerance $\tau > 0$. Let $\alpha \gets \mathsf{QPStat}_{\mc{E}}(\rho, O, \tau)$ be the output of a quantum statistical query. Then, we assume that $\alpha$ satisfies
    \begin{equation}
        \mathbf{E}[\alpha] = \tr(O\mathcal{E}(\rho)),
    \end{equation}
    where the expectation is over the oracle's internal randomness.
\end{assumption}
In such cases, we can obtain a tighter query complexity upper bound to solve Problem \ref{prb:spt-avg}, and without the restriction that $\tau < \Tilde{\epsilon}$. We present this result in the following corollary.
\begin{corollary}
\label{cor:spt-avg-with-assumption}
    Under Assumption \ref{ass:qpstat-assumption} on the output of \textsf{QPStat}, \Cref{alg:avg-spt} succeeds at Problem \ref{prb:spt-avg} using
    \begin{equation}
    \label{eq:rig-guarantee-assumption}
    N = \tau^2 M \log(Mn/\delta) 2^{\mc{O}(\log (n) \log(1/\epsilon))}
    \end{equation}
    queries of tolerance $\tau$ to $\mathsf{QPStat}_{\mc{E}}$.
\end{corollary}
\begin{proof}
Under this assumption, we can eliminate the intermediate quantity $x_P^\prime(Q)$ in the previous proof, and directly apply Hoeffding's inequality to obtain the bound on $|\hat{x}_P(Q)-x_P(Q)|$. Note that for all $l \in [N]$, 
\begin{equation}
    |\tr(P\rho_l)y_l - \tr(P\rho_l) \tr(O \mc{E}\rho_l)| \leq \tau,
\end{equation}
and 
\begin{equation}
    \mathbf{E}_{\rho_l} [\tr(P\rho_l)y_l] = x_P(O).
\end{equation}
Thus, using $N = \mc{O}\left(\frac{\tau^2 \log((3n)^k/\delta)}{\Tilde{\epsilon}^2}\right)$, from Hoeffding's inequality, we see that for any $P \in \pauli_n, |P| \leq k$, with probability at least $1 - \delta/(3n)^k$,
\begin{equation}
\label{eq:alg-triangle-pt3}
    |x_P(O) - \hat{x}_P(O)| \leq \Tilde{\epsilon}.
\end{equation}
Using a union bound over failure probabilities, and noting that there are at most $(3n)^k$ Pauli operators of degree $k$, we see that \Cref{eq:alg-triangle-pt3} holds for all $P \in \pauli_n, |P| \leq k$. The stated query complexity can be obtained using \Cref{eq:alg-hyperparam}.
\end{proof}

\subsection{Numerical Simulations}
\label{sec:simulations}
In this section, we demonstrate the performance of \Cref{alg:avg-spt} through numerical simulations. The code for our simulations is available in a public Github repository \footnote{\url{https://github.com/chirag-w/qpsq-learning}}. Before presenting our simulations, we remark on our approach towards simulating the $\mathsf{QPStat}$ oracle. In order to construct the output of $\mathsf{QPStat}$, we assume the oracle uses a method to estimate the expectation value of an observable, such as the ones shown in \cite{classical_shadow_tomography, hadfield2020measurements_locally_biased, Huang_2021_derandom, Wu2023overlappedgrouping}. In our simulation, in order to emulate the behaviour of these methods, we compute the true expectation value and output the result after adding a normally distributed error to it. The error is sampled from a normal distribution such that it is within the specified tolerance with high probability. We note that our learning model puts no assumptions on the error in the output of the queries, and in theory, this error can come from any arbitrary distribution as long as it is within the tolerance with high probability. However, we will argue that this simple method of generating the sample data already captures these scenarios well enough for the purpose of our simulations. We also compare this method with the use of classical shadows for evaluating the outcome of an observable \cite{classical_shadow_tomography} in \Cref{fig:compare-cs-norm}. We see that for the same target tolerance and success probabilities, the classical shadow method produces less error than that generated using a normal distribution. In fact, the normally distributed error achieves the exact value for the success probability, while any practical method would produce the same error or less, as it would come with a potentially looser bound. Thus, our emulated oracle would produce greater deviations than in a practical implementation, implying that the real-life performance of the algorithm would only be similar to or better than the simulations. We use these emulated oracles for the simulation of the learning algorithms.
\begin{figure}[H]
\centering
\includegraphics[width=0.5\textwidth]{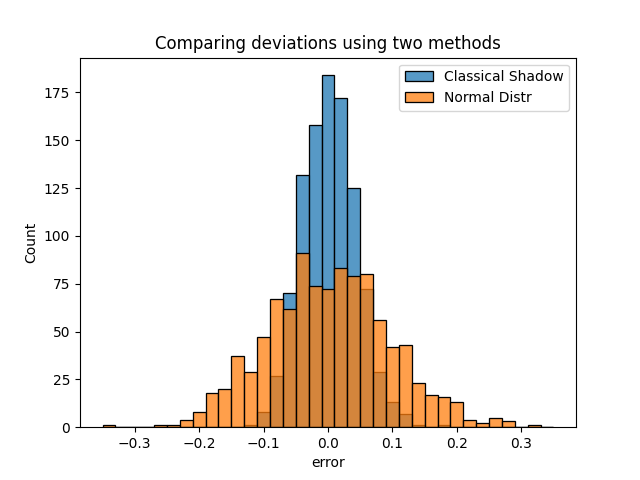}
\caption{\small Comparison between simulated errors generated from a normal distribution and those generated using classical shadow tomography to evaluate the EV of the Pauli-Z observable on random single-qubit stabilizer states after evolution under a fixed haar-random unitary. We fix a tolerance value $\tau = 0.2$, and the probability of the deviation lying outside the tolerance, $\delta = 0.0455$}
\label{fig:compare-cs-norm}
\end{figure}

In \Cref{fig:learn-arb-alg-haar}, we show the simulated performance of \Cref{alg:avg-spt} in predicting properties of 10 haar-random unitaries over 6 qubits for a range of tolerances. We consider $O = Z \otimes I ... \otimes I$, the Pauli-Z observable on the first qubit. We consider three distributions of target states, namely the uniform distributions over the computational basis states, the stabilizer product states and haar-random states. We can see from \Cref{fig:learn-arb-alg-haar} that a lower QPSQ tolerance results in a lower prediction error for the same number of queries. We also see that the algorithm achieves similar performance when predicting the outcome on computational basis and stabilizer product states, even though the uniform distribution over the computational basis states is not locally flat and thus outside the performance guarantee. On the other hand, the distribution over haar-random states is within the guarantee, and the algorithm performs best on this distribution.
\begin{figure}
    \vspace{-1em}
     \centering
     \begin{subfigure}[b]{0.42\textwidth}
         \centering
         \includegraphics[width=\textwidth]{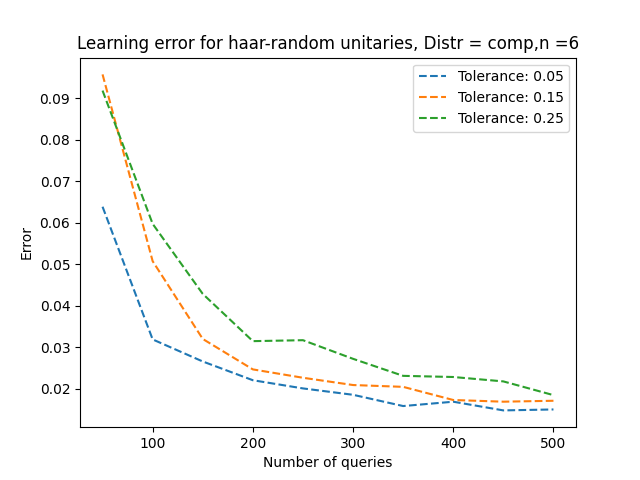}
         \caption{Computational basis states}
         \label{fig:haar-error-classical}
     \end{subfigure}
     \hfill
     \begin{subfigure}[b]{0.42\textwidth}
         \centering
         \includegraphics[width=\textwidth]{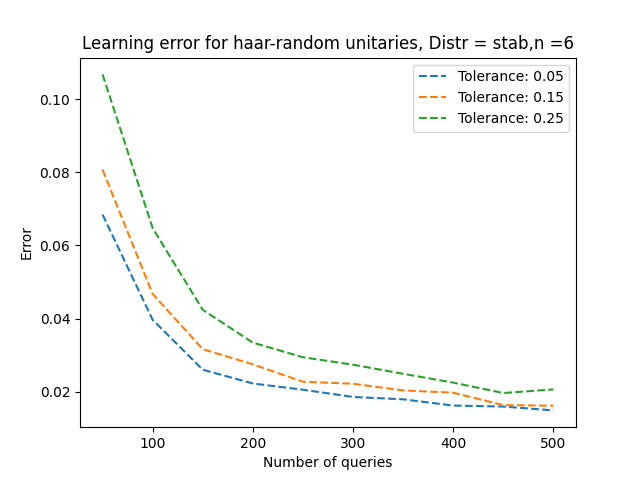}
         \caption{Stabilizer states}
         \label{fig:haar-error-stab}
     \end{subfigure}
     \hfill
     \begin{subfigure}[b]{0.42\textwidth}
         \centering
         \includegraphics[width=\textwidth]{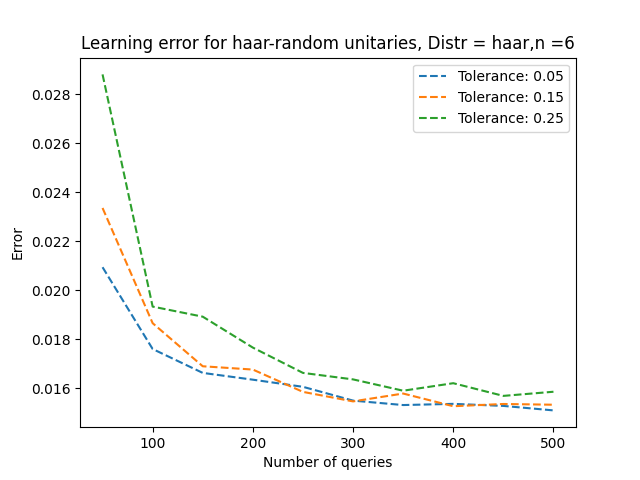}
         \caption{Haar-random states}
         \label{fig:haar-error-haar}
     \end{subfigure}
        \caption{Average performance of the learning algorithm on 10 haar-random 6-qubit unitaries, in predicting the outcome of $Z_1$ on three target distributions}
        \label{fig:learn-arb-alg-haar}
    \vspace{-1em}
\end{figure}
\section{Lower Bounds for Diamond Distance Learning}\label{sec:diamond}

In this section, rather than predicting properties of quantum processes, we consider the harder problem of learning them within diamond distance from $\mathsf{QPStat}$ queries. Here, we provide average-case query complexity lower bounds for learning unitary $2$-designs and Haar-random unitaries. We start by stating our lower bound for average-case learning exact $2$-designs (\Cref{def:designs}).
\begin{theorem}[Lower bound for learning exact unitary 2-designs]
    \label{thm:hardness-exact-design}
     Let $0 < \tau \leq \epsilon, \epsilon+2\tau < 1 - 1/2^n,$ and $ \mu_H^{(2)}$  be an exact unitary 2-design over $n$-qubits. Let there be an algorithm that, given access to $\mathsf{QPStat}_U$ for some $U \sim \mu_H^{(2)}$, outputs a quantum channel $\Phi$ with $d_\diamond(\mc{U},\Phi) \leq \epsilon$ with probability $\alpha$ over its internal randomness and probability $\beta$ over $U \sim \mu_H^{(2)}$. Then, such an algorithm must make $q$ queries of tolerance $\tau$ to $\mathsf{QPStat}_U$, where
    \begin{equation}
        q+1 \geq (2\alpha - 1)\beta \tau^2 (2^n+1).
    \end{equation}
\end{theorem}
\begin{remark}
    As random Clifford circuits are known to form unitary 3-designs \cite{webb2015clifford,zhu2017multiqubit}, \Cref{thm:hardness-exact-design} implies an exponential lower bound for learning random Cliffords from $\mathsf{QPStat}$ queries. Meanwhile, Clifford circuits can be learned efficiently using black-box access \cite{low2009learning,lai2022learning}. This gives us an exponential separation between QPSQ learners and those in the standard setting.
\end{remark}
 We now state our lower bound for learning additive approximate unitary $2$-designs (\Cref{def:designs}).
\begin{theorem}[Lower bound for learning approximate unitary 2-designs]
    \label{thm:hardness-approx-design}
    Let $0 < \tau \leq \epsilon < 1, \epsilon+2\tau < 1 - 1/2^n, 0 \leq \delta$, and $ \mu^{(2,\delta)}_H$ be a $\delta$-approximate unitary $2$-design and exact unitary $1$-design over $n$-qubits. Let there be an algorithm that, given access to $\mathsf{QPStat}_U$ for some $U \sim \mu^{(2,\delta)}_H$, outputs a quantum channel $\Phi$ with $d_\diamond(\mc{U},\Phi) \leq \epsilon$ with probability $\alpha$ over its internal randomness and probability $\beta$ over $U \sim \mu^{(2,\delta)}_H$. Then, such an algorithm must make $q$ queries of tolerance $\tau$ to $\mathsf{QPStat}_U$, where
    \begin{equation}
        q+1 \geq (2\alpha - 1)\beta \tau^2 \Omega( \min\{2^n, 1/\delta\}).
    \end{equation}
\end{theorem}
\begin{remark}
    For our lower bound on approximate $2$-designs, \Cref{thm:hardness-approx-design} additionally assumes that the measure is an exact $1$-design. This assumption is satisfied in practice by brickwork random quantum circuits, which form exact $1$-designs at any circuit depth and approximate $2$-designs at depths $\mc{O}(n + log(1/\delta))$ \cite{Haferkamp_2021}. This restriction can also be lifted in some cases. One can see that the proof of \Cref{thm:hardness-approx-design} also holds when the measure is a $\delta^\prime$-approximate $1$-design with $\delta^\prime < \tau/c$ for some absolute constant $c > 1$, but we do not explicitly discuss this case here. 
\end{remark}
 \Cref{thm:hardness-exact-design} already implies an exponential lower bound for learning Haar-random unitaries. However, using stronger concentration properties of the Haar measure, we show a doubly exponential lower bound in the following theorem.
\begin{theorem}[Lower bound for learning Haar-random unitaries]
    \label{thm:hardness-haar-measure}
     Let $0 < \tau \leq \epsilon < 1, \epsilon+2\tau < 1 - 1/2^n,$ and $ \mu_H$  be the unitary Haar-measure over $n$-qubits. Let there be an algorithm that, given access to $\mathsf{QPStat}_U$ for some $U \sim \mu_H$, outputs a quantum channel $\Phi$ with $d_\diamond(\mc{U},\Phi) \leq \epsilon$ with probability $\alpha$ over its internal randomness and probability $\beta$ over $U \sim \mu_H$. Then, such an algorithm must make $q$ queries of tolerance $\tau$ to $\mathsf{QPStat}_U$, where
    \begin{equation}
        q + 1 \geq (\alpha - 1/2)\beta \exp \left(\frac{2^n\tau^2}{48}\right). 
    \end{equation}
\end{theorem}
To prove these lower bounds, we first describe a many-vs-one distinguishing task and lower bound its $\mathsf{QPStat}$ query complexity in \Cref{sec:hardness-distinguishing}. Then, in \Cref{sec:learn-to-distinguish}, we show that for any measure over unitaries, average-case learning is at least as hard as a certain instance of the distinguishing task. Finally, we conclude the proofs of our lower bounds by computing the specific query complexities for this distinguishing task in \Cref{sec:hardness-designs,sec:hardness-haar}. Our techniques are similar to those used in other quantum statistical query lower bounds \cite{nietner2023average,arunachalam2023role,nietner2023unifying}, which are in turn inspired by lower bounds in the classical SQ setting \cite{feldman2017general}.

\subsection{Many-vs-One Query Complexity} \label{sec:hardness-distinguishing}
In this section, we will lower bound the QPSQ complexity of a many-vs-one distinguishing task for quantum channels. Specifically, consider a class of quantum channels $\mc{C}$ and another reference quantum channel $\Phi \notin C$. Then, given access to an unknown channel $\mc{E}$ through $\mathsf{QPStat}_\mc{E}$, one must determine whether $\mc{E} \in \mc{C}$, or if $\mc{E} = \Phi$, promised that one of these is the case. We provide the query complexity lower bound for this task in the following lemma.

\begin{lemma}[Many-vs-one distinguishing query complexity]
\label{lem:hardness-distinguishing}
    Let $\mc{C}$ be a class of quantum channels from $\mc{M}_{N,N}$ to $\mc{M}_{N,N}$, $\Phi \notin \mc{C}$ be another quantum channel from $\mc{M}_{N,N}$ to $\mc{M}_{N,N}$. Then, given access to $\mathsf{QPStat}_\mc{E}$, any algorithm that can distinguish, with probability $\alpha$ over its internal randomness, between $\mc{E} \in \mc{C}$ and $\mc{E} = \Phi$, for any measure $\mu$ over $\mc{C}$, must make at least
    \begin{equation}
        q \geq \frac{2\alpha - 1}{\max_{\rho, O} \mathbf{Pr}_{\mc{E}\sim \mu}\left[ | \tr(O \mc{E}(\rho)) - \tr(O \Phi(\rho)) | > \tau \right]}
    \end{equation}  
    queries of tolerance $\tau$ to $\mathsf{QPStat}_{\mc{E}}$.
\end{lemma}
\begin{proof}
   Let $\mc{A}$ be the algorithm that distinguishes between $\mc{C}$ and $\Phi$. Suppose $\mc{A}$ makes $q$ $\mathsf{QPStat}$ queries $\{(\rho_i, O_i)\}_{i \in [q]}$ based on its internal randomness and the previous responses. Let $p_{d}$ be the probability that $\mc{A}$ makes a distinguishing query, i.e.
   \begin{equation}
       p_d = \mathbf{Pr}_{\mc{A}, \mc{E} \sim \mu}\left[ \exists i \in [q] : |\tr(O_i\mc{E}(\rho_i)) - \tr(O_i\Phi(\rho_i))| > \tau \right]. 
   \end{equation}
   Now, suppose $\mc{A}$ receives $\tr(O_i\Phi(\rho_i))$ as responses for all $q$ queries. When $\mc{E} \in \mc{C}$, this happens with probability at most $1-p_d$. In this case, as the queries are consistent with $\mc{E} = \Phi$, by the correctness of $\mc{A}$, we obtain ``$\mc{E} \in \mc{C}$'' with probability at most $1-\alpha$. Thus, by the correctness of $\mc{A}$, we have
   \begin{equation}
       \alpha \leq p_d \cdot 1 + (1-p_d)(1-\alpha) \leq 1 - \alpha +p_d,
   \end{equation}
   where the second inequality uses that $\alpha \leq 1$.
   Rearranging the terms, we obtain
   \begin{equation}
   \label{eq:pdist-lower}
       p_d \geq 2\alpha -1.
   \end{equation}
   We can now upper bound $p_d$ using the union bound.
   \begin{align}
       p_d &= \mathbf{Pr}_{\mc{A}, \mc{E} \sim \mu}\left[ \exists i \in [q] : |\tr(O_i\mc{E}(\rho_i)) - \tr(O_i\Phi(\rho_i))| > \tau \right]
       \\& \leq \sum_{i \in [q]} \mathbf{Pr}_{\mc{A}, \mc{E} \sim \mu}\left[|\tr(O_i\mc{E}(\rho_i)) - \tr(O_i\Phi(\rho_i))| > \tau \right]
       \\&\leq q \max_{\rho, O} \mathbf{Pr}_{\mc{E} \sim \mu}\left[|\tr(O\mc{E}(\rho)) - \tr(O\Phi(\rho))| > \tau \right].
   \end{align}
   Together with \Cref{eq:pdist-lower}, we obtain the desired result.
\end{proof}

\subsection{Average-Case Learning Query Complexity} \label{sec:learn-to-distinguish}
In this section, we provide an average-case query complexity lower bound for learning unitaries from some measure. In \Cref{lem:worst-case-learning-to-distinguishing}, we show that learning a class of unitaries is as hard as a certain instance of the many-vs-one distinguishing task of \Cref{sec:hardness-distinguishing}. We then combine \Cref{lem:hardness-distinguishing,lem:worst-case-learning-to-distinguishing} to obtain our average-case learning lower bound in \Cref{lem:hardness-learning-avg}.

Given an algorithm for learning classes of quantum states and distributions from QSQs, \cite{nietner2023average,arunachalam2023role,nietner2023unifying} show that a single additional query suffices to distinguish this class from any reference object that is sufficiently far. However, such a result does not hold in general when using $\mathsf{QPStat}$ queries for distinguishing channels far in diamond distance. This is due to the fact that the optimal distinguishing state may need to be prepared over an additional ancillary register, and such a query cannot be made with our chosen definition of the $\mathsf{QPStat}$ oracle.

While we are unable to show a general reduction from distinguishing to learning with a single query, we will look at the specific case when the class of channels consists only of unitary channels and the fixed reference object is the depolarizing channel. In particular, we show that learning a class of unitaries from QPSQs is as hard as distinguishing it from the depolarizing channel. 

\begin{lemma}[Learning unitaries is as hard as distinguishing them from $\phidep$.]
\label{lem:worst-case-learning-to-distinguishing}
    Let $0 < \tau \leq \epsilon < 1$, $\epsilon+2\tau < 1 - \frac{1}{2^n}$. Let $\mc{C} \subseteq \mc{U}_{2^n}$ be a class of unitary channels. Let $\mc{A}$ be a learning algorithm that, given access to $\mathsf{QPStat}_U$ for some $U \in \mc{C}$, and with probability $\alpha$ over its internal randomness, outputs a channel $\Phi$ such that $d_\diamond(\mc{U}, \Phi) \leq \epsilon$, using $q_L$ queries of tolerance $\tau$ to $\mathsf{QPStat}_U$. Then, 
    \begin{equation}
        q_L + 1 \geq q_D,
    \end{equation}
    where $q_D$ is the number of $\mathsf{QPStat}$ queries of tolerance $\tau$ necessary to distinguish $\mc{C}$ from the maximally depolarizing channel $\phidep$ with probability $\alpha$.
\end{lemma}
\begin{proof}
    We will use $\mc{A}$ as a subroutine and construct a distinguisher making at most one additional query. We first run $\mc{A}$ on the unknown channel $\mc{E} \in \mc{C} \cup \{\phidep\}$. If the output of the algorithm is not a valid quantum channel, we output ``$\mc{E} = \phidep$''. Otherwise, let the output be a quantum channel $\Phi$. Now, we check if $\min_{\mc{E} \in \mc{C}} d_\diamond(\mc{E}, \Phi) \leq \epsilon$, and output ``$\mc{E} = \phidep$'' if not. While this step may be computationally expensive, it does not incur any additional queries.
    
    Now, the output of $\mc{A}$ will be some quantum channel $\Phi$. Then, we query $\mathsf{QPStat}_{\mc{E}}$ with $\rho = \ketbra{0}{0}$ and $O = \Phi( \ketbra{0}{0})$. Let $v \gets \mathsf{QPStat}_{\mc{E}} (\rho, O, \tau)$. If $v \geq 1 - \epsilon - \tau$, output ``$\mc{E} \in \mc{C}$''. Otherwise, output ``$\mc{E} = \phidep$''.
    
    To prove the correctness of this algorithm in distinguishing $\mc{C}$ from $\phidep$, we will condition on the success of $\mc{A}$, which occurs with probability $\alpha$. Now, when $\mc{E} \in \mc{C}$, let $U$ be the unitary corresponding to $\mc{E}$. By the correctness of $\mc{A}$, we obtain a quantum channel $\Phi$ with $d_\diamond(\mc{E},\Phi) \leq \epsilon$, so we proceed to the next step and make the $\mathsf{QPStat}$ query.  We can now lower bound $v$ as follows:
    \begin{align}
        v &\geq \tr(OU\rho U^\dagger) - \tau
        \\& = \tr(\Phi(\ketbra{0}{0}) U \ketbra{0}{0} U^\dagger) - \tau
        \\&= F(\Phi(\ketbra{0}{0}), U \ketbra{0}{0} U^\dagger) - \tau
        \\&\geq 1 - d_{\mathrm{tr}}(\Phi(\ketbra{0}{0}), U \ketbra{0}{0} U^\dagger) - \tau
        \\&\geq 1 - d_\diamond(\Phi, \mc{E}) - \tau
        \\&\geq 1 - \epsilon - \tau,
    \end{align}
    where the fourth line makes use of \Cref{ref:eq-dtr-infidel}. Thus, our distinguisher will correctly output ``$\mc{E} \in \mc{C}$''.
    
    In the second case, when $\mc{E} = \phidep$, the action of $\mc{A}$ may not be well defined or $\mc{A}$'s output may be far from all channels $\mc{C}$. In this case, our distinguisher correctly outputs ``$\mc{E} = \phidep$''. On the other hand, if $\mc{A}$'s output is some channel $\Phi$ close to some channel in $\mc{C}$, we make the $\mathsf{QPStat}$ query. Then, we can upper bound $v$ as follows:
    \begin{align}
        v &\leq \tr(O \phidep(\rho)) + \tau
        \\&= \tr(\Phi(\ketbra{0}{0}) \frac{I}{2^n}) + \tau
        \\&= \frac{1}{2^n} + \tau.
    \end{align}
    By assumption, $\epsilon+2\tau < 1 - \frac{1}{2^n}$. Thus, in this case, $v < 1 - \epsilon - \tau$ and our distinguisher correctly outputs ``$\mc{E} = \phidep$''.
    
    As the distinguisher succeeds conditioned on the success of $\mc{A}$, it has an overall success probability of at least $\alpha$. Further, the distinguisher makes at most $q_L+1$ queries. This concludes the proof.
\end{proof}
Now, we extend the above result to an average-case learning lower bound for learning unitaries.
\begin{lemma}[Average-case lower bound]
    \label{lem:hardness-learning-avg}
     Let $0 < \tau \leq \epsilon,\epsilon+2\tau < 1 - \frac{1}{2^n}$, $\mathcal{C} \subseteq \mc{U}_{2^n}$ be a class of unitaries, and $\mu$ some measure over $\mathcal{C}$. Let there be a learning algorithm that, given access to $\mathsf{QPStat}_U$ for some $U \in \mc{C}$, outputs a quantum channel $\Phi$ such that $d_\diamond(\mc{U},\Phi) \leq \epsilon$ with probability $\alpha$ over its internal randomness and probability $\beta$ over the random unitary $U \sim \mu$. Then, such an algorithm must make $q$ queries of tolerance $\tau$ to $\mathsf{QPStat}_U$, where
     \begin{equation}
     \label{eq:hardness-learning-avg}
         q + 1 \geq \frac{(2\alpha - 1)\beta}{\max_{\rho, O} \mathbf{Pr}_{\mc{E}\sim \mu}\left[| \tr(O \mc{E}(\rho)) - \tr(O \phidep(\rho)) | > \tau \right]}.
     \end{equation}
\end{lemma}
\begin{proof}
    Consider the subset $\mc{C}^\prime \subseteq \mc{C}$ of measure $\mu(\mc{C}^\prime) = \beta$ on which the learning algorithm succeeds. Let $\Tilde{\mu}$ be the measure conditioned on $\mc{C}^\prime$, i.e. $\Tilde{\mu}(U) = \mu(U | U \in \mc{C}^\prime)$. Then, for any $\rho,O$,
    \begin{align}
        \mathbf{Pr}_{U \sim \Tilde{\mu}} \left[ |\tr(OU\rho U^\dagger) - \tr(O \phidep(\rho))| > \tau\right] &= \mathbf{Pr}_{U \sim \mu} \left[ |\tr(OU\rho U^\dagger) - \tr(O \phidep(\rho))| > \tau | U \in \mc{C}^\prime \right] \\& \leq 
        \frac{\mathbf{Pr}_{U \sim \mu} \left[ |\tr(OU\rho U^\dagger) - \tr(O \phidep(\rho))| > \tau \right]}{\mathbf{Pr}_{U \sim \mu}\left[U \in \mc{C}^\prime\right]}
        \\&= \label{eq:prob-avg-case}\frac{\mathbf{Pr}_{U \sim \mu} \left[ |\tr(OU\rho U^\dagger) - \tr(O \phidep(\rho))| > \tau \right]}{\beta}.
    \end{align}
    The average-case learning algorithm with respect to $\mu$ implies a worst-case learner for the class $\mc{C}^\prime$. We can then invoke \Cref{lem:worst-case-learning-to-distinguishing} to lower bound the complexity for this task. The corresponding distinguishing task is $C^\prime$ against $\phidep$ with respect to the measure $\Tilde{\mu}$ over $C^\prime$, and its query complexity can be lower bounded using \Cref{lem:hardness-distinguishing} and \Cref{eq:prob-avg-case}. This concludes the proof.
\end{proof}
\subsection{Proofs of \Cref{thm:hardness-exact-design,thm:hardness-approx-design}} \label{sec:hardness-designs}
In this section, we detail the proofs of \Cref{thm:hardness-exact-design,thm:hardness-approx-design}. Using \Cref{lem:hardness-learning-avg}, it suffices to bound the probability in the denominator of \Cref{eq:hardness-learning-avg} for exact and approximate unitary $2$-designs. To this end, we first bound the variance of $\tr(O U \rho U^\dagger)$ for exact $2$-designs.

\begin{lemma}[Variance over $2$-designs is exponentially small]
\label{lem:haar-var}
    Let $\rho \in \mathcal{S}_N, O \in \mathcal{M}_{N,N}, \norm{O}_\infty \leq 1$ and $\mu_H^{(2)}$ be a unitary $2$-design. Then, 
    \begin{equation}
        \underset{U \sim \mu_H^{(2)}}{\mathbf{Var}} \tr (OU\rho U^\dag) \leq \frac{1}{N+1}.
    \end{equation}
\end{lemma}
\begin{proof}
    As the variance is a second-order moment, we can equivalently bound the variance over the unitary Haar-measure $\mu_H$. From Lemma~\ref{lem:haar-moments}, we have
    \begin{equation}
        \underset{U \sim \mu_H}{\mathbf{E}} \tr(OU\rho U^\dag) = \frac{\tr(O)}{N}.
    \end{equation}
    Again, from Lemma~\ref{lem:haar-moments},
        \begin{align}
         \underset{U \sim \mu_H}{\mathbf{E}}\left[\tr(OU \rho U^\dag)^2\right] &= \underset{U \sim \mu_H}{\mathbf{E}}\left[\tr(O^{\otimes 2}U^{\otimes 2} \rho^{\otimes 2} U^{\dag \otimes 2})\right] \\&= \tr(O^{\otimes 2} \underset{U \sim \mu_H}{\mathbf{E}}[U^{\otimes 2} \rho^{\otimes 2} U^{\dag \otimes 2}]) \\&= \tr(O^{\otimes 2} (c_{\mathbb{I}, \rho^{\otimes 2}}\mathbb{I} + c_{\mathbb{F}, \rho^{\otimes 2}}\mathbb{F})) \\&= c_{\mathbb{I}, \rho^{\otimes 2}}\tr(O)^2 + c_{\mathbb{F}, \rho^{\otimes 2}}\tr(O^2).
        \end{align}
    Now, 
    \begin{align}
        c_{\mathbb{I}, \rho^{\otimes 2}} &= \frac{N\tr(\rho^{\otimes 2})- \tr(\rho^{\otimes 2}\mathbb{F})}{N(N^2-1)} \\&= \frac{N\tr(\rho)^2 - \tr(\rho^2)}{N(N^2-1)} \\&= \frac{N - \tr(\rho^2)}{N(N^2-1)},
    \end{align}
    and,
    \begin{align}
        c_{\mathbb{F}, \rho^{\otimes 2}} &= \frac{N\tr(\mathbb{F}\rho^{\otimes 2}) - \tr(\rho^{\otimes 2})}{N(N^2-1)} \\&= \frac{N\tr(\rho^2)-1}{N(N^2-1)} \\& \leq \frac{N-1}{N(N^2-1)}  \\&= \frac{1}{N(N+1)},
    \end{align}
    where the inequality follows from the fact that $\tr(\rho^2) \leq 1$. 
    
    Finally,
    \begin{align}
        \underset{U \sim \mu_H}{\mathbf{Var}}\left[\tr(OU \rho U^\dag)\right] &= c_{\mathbb{I}, \rho^{\otimes 2}}\tr(O)^2 + c_{\mathbb{F}, \rho^{\otimes 2}}\tr(O^2) - \frac{\tr(O)^2}{N^2}
        \\& \leq \left(\frac{N - \tr(\rho^2)}{N(N^2-1)} - \frac{1}{N^2}\right) \tr(O)^2 + \frac{\tr(O^2)}{N(N+1)} \\& = \frac{1-N\tr(\rho^2)}{N^2(N^2-1)} \tr(O)^2 + \frac{Tr(O^2)}{N(N+1)} \\& \leq \frac{\tr(O^2)}{N(N+1)} \\& \leq \frac{1}{N+1},
    \end{align}
    where the second inequality uses that $\tr(\rho^2) \geq 1/N$ and the last inequality is due to $\norm{O}_\infty \leq 1 \Rightarrow \tr(O^2) \leq N$.
\end{proof}
 We are now in a position to prove \Cref{thm:hardness-exact-design}.
\begin{proof}[Proof of \Cref{thm:hardness-exact-design}]
    We note that $\phidep$ is the expected channel over the unitary Haar measure. As this is a first-order moment, $\phidep$ is also the expected channel over unitary $2$-designs. We can thus use Chebyshev's inequality to obtain,
    \begin{equation}
        \underset{\rho,O}{\max} \underset{U \sim \mu_H^{(2)}}{\mathbf{Pr}} \left(\left|\tr(OU\rho U^\dag) - \left[\tr(O\phidep(\rho))\right]\right| > \tau \right) \leq \underset{\rho,O}{\max} \frac{1}{\tau^2} \underset{U \sim \mu_H^{(2)}}{\mathbf{Var}} \left[\tr(OU\rho U^\dag)\right].
    \end{equation}
    The desired result then follows by combining \Cref{lem:hardness-learning-avg,lem:haar-var}.
\end{proof}
Before we prove \Cref{thm:hardness-approx-design}, we show an upper bound on the variance of $\tr(OU \rho U^\dagger)$ over approximate unitary $2$-designs.

\begin{lemma}[Variance over approximate $2$-designs]
\label{lem:approx-variance}
    Let $\rho \in \mathcal{S}_N, O \in \mathcal{M}_{N,N}, \norm{O}_\infty \leq 1$ and $\mu_H^{(2,\delta)}$ be a $\delta$-approximate unitary $2$-design. Then,
       \begin{equation}
           \underset{U \sim \mu_{H}^{2,\delta}}{\mathbf{Var}} \left(\tr(OU\rho U^\dag)\right) \leq 
           \frac{1}{N+1}+ 3\delta.
       \end{equation} 
\end{lemma}
\begin{proof}
    We will bound the difference in variances over the Haar measure and approximate $2$-designs. For concise notation, let $f_{U}(\rho,O) = \tr(OU\rho U^\dag)$. Then, the difference in variances can be bounded as follows.
        \begin{align}
        \label{eq:approx-design-lemma-eq}
            \left| \underset{\mu_{H}^{2,\delta}}{\mathbf{Var}} \left(f_{U}(\rho,O)\right) - \underset{\mu_{H}}{\mathbf{Var}} \left(f_{U}(\rho,O)\right) \right| &\leq \left|\mathbf{E}{\mu_{H}^{2,\delta}} f_{U}(\rho,O)^2 - \mathbf{E}_{\mu_{H}} f_{U}(\rho,O)^2\right| \nonumber \\&+ \left|(\mathbf{E}_{\mu_{H}^{2,\delta}} f_{U}(\rho,O))^2 - (\mathbf{E}_{\mu_{H}} f_{U}(\rho,O))^2\right|,
        \end{align}
        where the inequality follows from the triangle inequality. We bound the first term as follows:
    \begin{align}
        \left|\mathbf{E}_{\mu_{H}^{2,\delta}} f_{U}(\rho,O)^2 - \mathbf{E}_{\mu_{H}} f_{U}(\rho,O)^2\right| &= \left| \tr\left(O^{\otimes 2}\left(\mathcal{M}^{(2)}_{\mu_H^{2,\delta}} - \mathcal{M}^{(2)}_{\mu_H}\right) (\rho^{\otimes 2}) \right)\right|
        \leq \delta,
    \end{align}
    where the inequality holds since $\mu_H^{2,\delta}$ is an additive $\delta$-approximate unitary 2-design. We now bound the second term in \Cref{eq:approx-design-lemma-eq}.
    \begin{align}
        \left|(\mathbf{E}_{\mu_{H}^{2,\delta}} f_{U}(\rho,O))^2 - (\mathbf{E}_{\mu_{H}} f_{U}(\rho,O))^2\right| &= |\mathbf{E}_{\mu_{H}^{2,\delta}} f_{U}(\rho,O) + \mathbf{E}_{\mu_{H}} f_{U}(\rho,O)| \nonumber \\&\cdot|\mathbf{E}_{\mu_{H}^{2,\delta}} f_{U}(\rho,O) - \mathbf{E}_{\mu_{H}} f_{U}(\rho,O)|
        \\& \leq 2 \cdot \left| \tr\left(O\left(\mathcal{M}^{(1)}_{\mu_H^{2,\delta}} - \mathcal{M}^{(1)}_{\mu_H}\right) (\rho)\right) \right|
        \\& \leq 2\delta.
    \end{align}
    In the first inequality, we use that $|f_U(\rho,O)| \leq 1$ when $\|O\|_\infty \leq 1$, and the last inequality holds since $\mu_H^{2,\delta}$ is an additive $\delta$-approximate unitary 2-design. \Cref{lem:haar-var} and the bounds on the two terms in \Cref{eq:approx-design-lemma-eq} suffice to prove the lemma.
\end{proof}
 Now, we are in a position to prove \Cref{thm:hardness-approx-design}.
\begin{proof}[Proof of \Cref{thm:hardness-approx-design}]
    Here, we have the additional assumption that the approximate $2$-design is an exact $1$-design. Thus, $\phidep$ is also the expected channel over $\mu_H^{(2,\delta)}$. We can thus use Chebyshev's inequality to obtain,
    \begin{equation}
        \underset{\rho,O}{\max} \underset{U \sim \mu_H^{(2,\delta)}}{\mathbf{Pr}} \left(\left|\tr(OU\rho U^\dag) - \left[\tr(O\phidep(\rho))\right]\right| > \tau \right) \leq \underset{\rho,O}{\max} \frac{1}{\tau^2} \underset{U \sim \mu_H^{(2,\delta)}}{\mathbf{Var}} \left[\tr(OU\rho U^\dag)\right].
    \end{equation}
    From \Cref{lem:approx-variance}, we see that when $\delta = \mc{O}(1/2^n)$, the variance is $\mc{O}(1/2^n)$. On the other hand, when $\delta = \omega(1/2^n)$, the variance is $\mc{O}(\delta)$. Together with \Cref{lem:hardness-learning-avg}, we obtain the desired result.
\end{proof}
\subsection{Proof of \Cref{thm:hardness-haar-measure}} \label{sec:hardness-haar}
In this section, we prove \Cref{thm:hardness-haar-measure} on the hardness of learning Haar-random unitaries. Using \Cref{lem:hardness-learning-avg}, it again suffices to bound the probability in the denominator of \Cref{eq:hardness-learning-avg}. To do this, we will use the following concentration result from \cite{meckes2013spectral}.
\begin{lemma}[Concentration of Measure for Haar-random unitaries, Corollary 17 of \cite{meckes2013spectral}]
\label{lem:haar-conc}
    Let $f : \mc{M}_{N,N} \rightarrow \mathbb{R}$ be a function from $N \times N$-dimensional matrices to the real numbers that is Lipschitz with Lipschitz constant $L$ with respect to the Schatten-2 norm. Let $\mu_H$ be the unitary Haar-measure. For any $\tau > 0$,
    \begin{equation}
        \underset{U \sim \mu_H}{\mathbf{Pr}} \left( 
                f(U) - \underset{U \sim \mu_H}{\mathbf{E}}[f(U)] > \tau
        \right) \leq \exp\left(\frac{-N\tau^2}{12L^2}\right)
    \end{equation}
\end{lemma}
We are now in a position to prove \Cref{thm:hardness-haar-measure}.
\begin{proof}[Proof of \Cref{thm:hardness-haar-measure}]
We will start by considering $f_{\rho,O}(U) = \tr(OU\rho U^\dagger)$. For any two unitaries $U,V$, 
\begin{align}
|f_{\rho,O}(U) - f_{\rho,O}(V)| &= |\tr(OU\rho U^\dagger) - \tr(OV\rho V^\dagger)|
\\&= \left|\tr(\rho U^\dagger OU) - \tr(\rho V^\dagger OV)\right|
\\&= \left|\tr(\rho  (U^\dagger OU - V^\dagger OV))\right|
\\&\leq \|\rho\|_2 \|U^\dagger OU - V^\dagger OV\|_2
\\& \leq \|U^\dagger OU - V^\dagger OU\|_2 + \|V^\dagger OU - V^\dagger OV\|_2
\\& \leq \|U^\dagger - V^\dagger\|_2 + \|U-V\|_2
\\& = 2\|U-V\|_2.
\end{align}
Here, $\|\cdot\|_2$ denotes the Schatten 2-norm. In the second and third lines we use cyclicity and linearity of the trace respectively, in the fourth line we use $\tr(A^\dagger B) \leq \|A\|_2 \|B\|_2$ (see \cite{baumgartner2011inequality}), in the fifth line we use the triangle inequality and the fact that $\|\rho\|_2 \leq 1$ for all states. Finally, in the second to last line, we use that $\|AB\|_2 \leq \|A\|_2 \|B\|_\infty$. Thus, the Lipschitz constant of $f_{\rho,O}$ is at most $2$.

Then, from Lemma \ref{lem:haar-conc}, and the fact that $\phidep$ is the expected channel over $\mu_H$, we obtain
\begin{equation}
    \underset{U \sim \mu_H}{\mathbf{Pr}} \left[\tr(OU\rho U^\dag) - \tr(O \phidep(\rho)) > \tau \right] \leq \exp\left( \frac{-2^n\tau^2}{48}\right).
\end{equation}
By analyzing $g_{\rho,O} = -f_{\rho,O}$, we can obtain a similar upper bound on the probability of $f_{\rho,O}(U)$ being significantly lower than its expectation. By taking a union bound, we obtain
\begin{equation}
    \underset{U \sim \mu_H}{\mathbf{Pr}} \left[ |\tr(OU\rho U^\dag) - \tr(O \phidep(\rho))| > \tau \right] \leq 2\exp\left( \frac{-2^n\tau^2}{48}\right).
\end{equation}
Together with \Cref{lem:hardness-learning-avg}, we obtain the desired result.
\end{proof}

\section{Application to CR-QPUF Security}
\label{sec:crqpuf}
In this section, we present applications of our QPSQ model in the realm of cryptography. Specifically, we focus on a particular class of hardware security primitives used for authentication known as Classical Readout Quantum Physically Unclonable Functions (CR-QPUFs). We demonstrate an attack against CR-QPUF-based authentication protocols using \Cref{alg:avg-spt}. As our attack has a quasi-polynomial time complexity, this does not break the formal security definition of CR-QPUFs. However, our approach highlights that any new, polynomial-time algorithm for shadow tomography of quantum processes from QPSQs would imply vulnerability for an appropriate class of CR-QPUFs.

Physical Unclonable Functions (PUFs) are hardware devices designed to resist cloning or replication, making them suitable for cryptographic tasks like authentication, identification, and fingerprinting \cite{ruhrmair2014pufs,chang2017retrospective,Id-QPUF,delavar2017puf}. Classically, PUFs have been realized using specific electrical circuits or optical materials \cite{pappu2002physical,guajardo2007fpga,gassend2002silicon,suh2007physical}. However, many of these implementations remain susceptible to various attacks, such as side-channel and machine-learning attacks \cite{ruhrmair2010modeling,ganji2016strong,tebelmann2019side,khalafalla2019pufs}. To overcome vulnerabilities due to ML-based attacks, Quantum PUFs were introduced and analyzed in \cite{QPUF}. However, the secure realization of Quantum PUFs shown in \cite{QPUF} requires implementing Haar-random unitaries, quantum communication and quantum memory, making them impractical in the near term. 

As a result, a range of PUF variations necessitating different levels of quantum capability have been proposed and explored. Among these are the Classical Readout Quantum Physically Unclonable Functions (CR-QPUFs), examined in \cite{CRQPUF-Original, CR-QPUF-single}. CR-QPUFs represent a middle ground between fully quantum PUFs and classical PUFs, aiming to achieve the security of Quantum PUFs while relying only on classical communication and storage. However, \cite{CR-QPUF-single} introduces a classical machine-learning attack directed at the initial proposal for CR-QPUFs and a slightly more sophisticated quantum circuit of the same type, demonstrating the vulnerability of such constructions. Nevertheless, because this construction is based on a relatively simple quantum circuit, it raises a significant open question in this domain: \emph{Does the susceptibility of CR-QPUFs to learning attacks stem from the simplicity of the quantum process itself, or does it arise fundamentally from the way CR-QPUFs have been defined?}
 
We partially address this question by showing that regardless of the underlying quantum process, CR-QPUFs relying on local measurements for authentication are vulnerable to a learning attack based on \Cref{alg:avg-spt}. First, we discuss the desired security properties and adversarial model for CR-QPUFs in \Cref{sec:crqpuf-sec}. Then, we outline an authentication protocol for CR-QPUFs in \Cref{sec:crqpuf-protocol}, and discuss the relation between CR-QPUFs and the $\mathsf{QPStat}$ oracle in \Cref{sec:crqpuf-qpsq}. We present our main result for this section in the form of a learning attack in \Cref{sec:crqpuf-vuln}.
 
\subsection{Security property and attack model} \label{sec:crqpuf-sec}
As discussed in \cite{QPUF}, the main cryptographic notion associated with PUFs in general, and particularly for the use-case of authentication and identification, is \emph{unforgeability}. Essentially, unforgeability means that an adversary having access to an efficient-size database of input-output pairs (also called challenge-response pairs (CRPs)) of a specific function of interest, should not be able to reproduce the output of the function on a \emph{new} input. This database is often obtained via the adversary querying an oracle that realises the black-box access to the function. Unforgeability is also extended to the quantum world generally in two directions: it can be defined for quantum processes (instead of classical functions) \cite{doosti2021unified} and also for quantum adversaries having access to quantum oracles of a classical function \cite{boneh2013quantum,alagic2018unforgeable,doosti2021unified}. As there are many notions of unforgeability involving quantum adversaries, and the details are outside the scope of this paper, we only present the vulnerability of the protocol through the lens of learnability. 

We consider two honest computationally bounded parties, a verifier $\mathcal{V}$ and a prover $\mathcal{P}$, communicating through a quantum or classical channel (depending on the challenge type). The purpose of the protocol is for the honest prover to prove their identity to the verifier with the promise that no quantum adversary can falsely identify themselves as the honest prover. A quantum adversary here is a quantum polynomial time (QPT) algorithm, that sits on the communication channel and can run arbitrary efficient quantum processes. The prover possesses a CR-QPUF device denoted as $\mathcal{C}$, which they will use for identification, and which is associated with a CPTP $\mathcal{E}$. The verifier on the other hand has a database of CRPs of $\mathcal{C}$, which is obtained by having direct access to $\mathcal{C}$ in the setup phase and recording the queries and their respective outputs. In this scenario, it is often considered that the device is later sent to the prover physically and from that point is possessed by the prover. To consider a stronger attack model, in addition to having access to the communication channel, we also assume the adversary has access to a polynomial-size database of CRPs. Here, the adversarial models are often categorized into different classes depending on the level of access assumed for the adversary to obtain such data. A \emph{weak adversary} only has access to a randomly selected set of CRPs, often obtained by recording the communications on the channel over a certain period of time. We will consider a stronger security model, where the adversary is \emph{adaptive} and can have oracle access to the device, i.e. issuing their desired queries to the CR-QPUFs. Although the adaptive adversarial model seems very strong, it is still realistic and often the desired security level when it comes to identification protocols and PUF-based schemes. A reason for this is that the device needs to be physically transferred at least once, which gives an adversary the chance to directly query and interact with it.  In the standard adversarial model, a QPT adversary $\mathcal{A}$ can collect polynomially-many CRPs, by issuing any desired state as input, and then use this dataset to learn $\mathcal{C}$. Then, in the challenge phase of the protocol, the adversary can provide a verifiable outcome for a new challenge state, and hence break the security. However, we note that the specific attack we demonstrate has a quasi-polynomial query and time complexity for providing such an outcome.
\subsection{The general structure of a CR-QPUF-based authentication protocol}
\label{sec:crqpuf-protocol}
We now define a general authentication protocol for CR-QPUFs and discuss how their security is defined. For the CR-QPUFs, we work with a definition similar to \cite{CR-QPUF-single} based on quantum statistical queries. While our definition differs slightly from that of 
\cite{CR-QPUF-single}, it is a natural extension of secure Quantum PUFs \cite{QPUF} to the classical readout setting. We will discuss the differences between these frameworks in \Cref{sec:crqpuf-qpsq}. For now, we denote a CR-QPUF as $\mathcal{C}$ and abstractly define it as a completely positive trace preserving (CPTP) map $\mathcal{E}$ over the $n$-qubit state space, able to produce statistical queries for any challenge in the challenge set, given an observable $O$ and a tolerance parameter $\tau$. We define the protocol between a verifier $\mathcal{V}$ and a prover $\mathcal{P}$ in Protocol \ref{prot:auth}.
\begin{protocol}
    \small
    \caption{The general structure of the CR-QPUF-based Authentication Protocol}
    \label{prot:auth}
    \begin{enumerate}
        \item \textbf{Setup phase:}
        \begin{enumerate}
            \item The Verifier $\mathcal{V}$ possesses a CR-QPUF $\mathcal{C}$ associated with the quantum process $\mathcal{E}$.
            \item The Verifier $\mathcal{V}$ and the Prover $\mathcal{P}$ agree on an observable $O$ and a threshold $\tau$.
            \item $\mathcal{V}$ builds a database $D$ of CRPs of $\mathcal{C}$ by querying $\mathsf{QPStat}_{\mathcal{E}}$ on $O$, with threshold $\tau$\\
            \textsf{For $i = 1$ to $N$}: For challenge $x_i$, a quantum state $\rho_i$ (or $\rho(x_i)$), is sampled from a selected distribution $\mathcal{D}$ and prepared as multiple copies, issued to $\mathsf{QPStat}_{\mathcal{E}}$, and the respective statistical query response $y_i$ is recorded. Thus the database $D = \{(x_i, y_i)\}_{i=1}^N$ is then constructed.
            
            \item $\mathcal{V}$ physically sends $\mathcal{C}$ to the Prover $\mathcal{P}$.

            At this point $\mathcal{V}$ possesses database $D$ and $\mathcal{P}$ possesses the device $\mathcal{C}$.
        \end{enumerate}
        \item \textbf{Authentication phase:}
        \begin{enumerate}
            \item $\mathcal{V}$ selects a CRP $(x_i,y_i)$ uniformly at random from $D$  and issues the challenge $x_i^t$ to $\mathcal{P}$.\\
            if $t=c$ the challenge is classical; if $t=q$ the challenge has been sent as multiple copies of the associate quantum state $\rho_i$.
            \item $\mathcal{P}$ receives the challenge  $x_i$ and proceeds as follows:\\
            \textsf{if t=c:} $\mathcal{P}$ creates $\mathsf{poly}(1/\tau)$ copies of the state $\rho_i = \rho(x_i)$.\\
            \textsf{else if t=q:} $\mathcal{P}$ proceeds to next step.
            \item $\mathcal{P}$ obtains the output of  the statistical query $y_i^\prime$ by issuing $\rho_i$ to $\mathcal{C}$. This step is similar to the setup phase.
            \item $\mathcal{P}$ sends $y_i^\prime$ to $\mathcal{V}$.
            \item $\mathcal{V}$ receives $y_i^\prime$ and verifies it. If $|y_i - y_i^\prime| \leq 2\tau$, the authentication is successfully passed. Otherwise, $\mathcal{V}$ aborts.
        \end{enumerate}
    \end{enumerate}
\end{protocol}

We also note that for a physical device such as a CR-QPUF, the statistical query oracle $\mathsf{QPStat}$ abstractly models a natural and physical interaction with the device, which is querying it with the given challenge and measuring the output quantum states on a desired observable and over several copies to estimate the expectation value of the observable. In other words, the oracle is the physical device itself and not a separate entity or implementation.

In the context of these authentication protocols utilizing CR-QPUFs, we encounter two main factors governing the complexity of the underlying components. Firstly, there is the complexity of the channel representing CR-QPUF. Secondly, there is the choice of the observable. In our work, we assume the protocol is carried out using observables that are efficiently estimable, taking into account practical scenarios where both the verifier and prover can effectively measure and estimate the CR-QPUF's outcome, regardless of the underlying circuit's complexity. As such, and given that we aim to provide attacks with provable guarantees, we assume that the observable $O$ selected during the setup phase is an efficient observable, i.e. we assume that $O$ has a polynomially bounded number of terms in its Pauli representation as well as a bounded number of qubits each term acts non-trivially on. This is a physically well-motivated assumption, as demonstrated in current state-of-the-art research on estimating the expectation value of an observable \cite{ Wu2023overlappedgrouping}. In \cite{Wu2023overlappedgrouping}, the authors provided a framework which unifies a number of the most advanced and commonly studied methods, such as those in \cite{classical_shadow_tomography, Huang_2021_derandom}. While this assumption covers a wide class of observables, there are some non-local observables that can be measured efficiently using more specific techniques, such as the one in \cite{observable-measurement-2}. Nevertheless, under this assumption, we are able to formally demonstrate the vulnerability of a very large class of CR-QPUF authentication protocols. However, that does not imply the security of the cases that might be excluded due to our assumption on the observable, and heuristic attacks might still be applicable to scenarios in which the protocol uses a complicated and highly non-local observable.

The \emph{correctness} or \emph{completeness} of the protocol, which is defined as the success probability of an honest prover in the absence of any adversary or noise over the channel, should be 1. The \emph{soundness} or \emph{security} of the protocol, ensures that the success probability of any adversary (depending on the adversarial model) in passing the authentication should be negligible in the security parameter. For protocols defined as above, the completeness is straightforward, hence we only define and discuss the soundness. 

\begin{definition}[Soundness (security) of the CR-QPUF-based Authentication Protocol] 
\label{def:sound-crqpuf}
We say the CR-QPUF-based authentication protocol~\ref{prot:auth} is secure if the success probability of any QPT adaptive adversary $\mathcal{A}$ in producing an output $\tilde{y}$ for any $x$ sampled at random from a database $D$, over a distribution $\mathcal{D}$, that passes the verification by satisfying the condition $|y - \tilde{y}| \leq 2\tau$, is negligible in the security parameter.
\end{definition}

\subsection{CR-QPUFs in QPSQ framework} \label{sec:crqpuf-qpsq}
\paragraph{CR-QPUF Model:} While we have been discussing the \textsf{QPStat} oracles abstractly so far, the CR-QPUF device must, in practice, be able to take a quantum state, an observable and a tolerance parameter $\tau$, and output with high-probability a $\tau$-estimate of the expectation value. The device can thus be modeled by the oracle \textsf{QPStat}$_\mathcal{E}$ for the fixed underlying channel $\mathcal{E}$. On receiving a query, the device would apply $\mathcal{E}$ to multiple copies of the state, estimate the expectation value of the observable, and respond with that value as output. There are multiple methods for this estimation, such as those shown in \cite{classical_shadow_tomography, hadfield2020measurements_locally_biased, Huang_2021_derandom, Wu2023overlappedgrouping}, usually requiring measurements on $\mathsf{poly}(1/\tau)$ copies of the state for statistical estimation. In a real implementation of the protocol, the required copies of the input state would either be received through a quantum communication channel or prepared by the device given a classical description as input.

Our definition of CR-QPUFs is similar to the one considered in prior work~\cite{CR-QPUF-single}, where the device was also modeled by a quantum statistical query oracle. While the two models are similar in spirit, we note that there are some differences between the two definitions. The main difference is that in \cite{CR-QPUF-single}, the challenges have been defined in the form of descriptions of unitaries instead of quantum states. Starting with a fixed state initialized as $\ketbra{0}{0}$, the input unitary $U_{\mathrm{in}}$ is applied on the noisy hardware, followed by repeated measurements in the computational basis. Finally, the mean of some statistical function is computed over the measurement results. The idea behind this kind of construction is that the unique noise fingerprint of the device may result in unforgeability. On the other hand, our model considers a device repeatedly implementing a fixed quantum process, that takes input states as challenges. Here, the implemented quantum process acts as a unique fingerprint of the device instead. Considering the setting of a fixed channel is also a natural extension of prior work on Quantum PUFs \cite{QPUF}.

\subsection{Vulnerability from learning results} \label{sec:crqpuf-vuln}
We are now ready to present our attack on Protocol~\ref{prot:auth}. To show the extreme case of our result, we can assume the quantum process $\mathcal{E}$ corresponds to a Haar-random unitary. We can also consider any arbitrary, fixed noise model on top of it to model the hardware-specific imperfections of the CR-QPUF. We use \Cref{alg:avg-spt} for our specific attack strategy. We construct the attack in \Cref{alg:attack-prot-auth}.

\begin{algorithm}
    \caption{QPSQ attack on Protocol \ref{prot:auth} with observable $O$, tolerance $\tau$}
    \label{alg:attack-prot-auth}
    \begin{algorithmic}
    \State \textbf{Setting hyperparameters:}
    \State $\epsilon \gets \tau$
    \State Set $N$ according to \Cref{cor:spt-avg-with-assumption} 
    \\
    \For{$i = 1$ to $N$} 
        \State $\rho_i \sim \stabn$
        \State Issue challenge $\rho_i$ to $\mathcal{C}$
        \State Receive response $y_i$
    \EndFor
    \State $S_N \gets \bigl\{(\rho_i,y_i)\bigr\}_{i = 1}^N$
    \State Learn $h$ according to \textbf{Learning} phase of \Cref{alg:avg-spt}
    \\
    \State \textbf{Forgery}
    \State Given challenge $x$ from $\mathcal{V}$, respond with $h(\rho(x))$
\end{algorithmic}
\end{algorithm}
Using \Cref{cor:spt-avg-with-assumption}, we present the performance guarantees of this algorithm in \Cref{thm:crqpuf-vuln}. Our result is valid for any challenge distribution $\mathcal{D}$ which satisfies the assumptions of our proposed algorithm (i.e. invariance under local Clifford operations). Two specific examples of such distributions are $\mathcal{D}_{Haar}$, consisting of Haar-random states over the $n$-qubit Hilbert space and $\mathcal{D}_{stab}$, uniformly random states from $\stabn$. In the first case the challenge states $\rho(x)$ are Haar-random states indexed by $x$ and in the second case, the challenge states are in the form of $\rho(x) = \bigotimes_{i = 1}^n \ket{\psi_{x^i}}\bra{\psi_{x^i}}$, where $x \in \{0,1,2,3,4,5\}^n$, and we have $\{\ket{\psi_0} = \ket{0}, \ket{\psi_1} = \ket{1}, \ket{\psi_2} = \ket{+}, \ket{\psi_3} = \ket{-}, \ket{\psi_4} = \ket{+i}, \ket{\psi_5} = \ket{-i}\}$. These two specific selections of distributions for challenge states give rise to two very distinct instances of the authentication protocol. In the first case where the challenges are selected from a Haar-random distribution, the challenge state is communicated through a quantum channel, in the form of multiple copies of the state $\rho(x)$, for the prover to be able to produce the response $y_x$. Intuitively we expect this to increase the security of the protocol since the adversary is unlikely to gain any information about the challenge state itself. However, this extra hiding comes with the price of generating and communicating $n$-qubit Haar-random states, which is often very resource-extensive. On the other hand, studying this case (especially when also considering the underlying process of $\mc{C}$ to be a Haar-random unitary) would be interesting because it would allow the comparison between a QPUF and a CR-QPUF with the same level of underlying resources. We note that Haar-random unitaries have been shown to satisfy the requirements of secure QPUFs~\cite{QPUF}. This highlights the importance of the type of challenge and the verification process in the security of these hardware-based protocols. Our result is formalized in the following theorem.
\begin{theorem}
\label{thm:crqpuf-vuln}  Under Assumption \ref{ass:qpstat-assumption} on the output of the CR-QPUF $\mc{C}$, for any underlying quantum process $\mc{E}$, any choice of observable $O$ given as the sum of few-body observables, and under any choice of challenge state sampled from a distribution $\mc{D}$ invariant under single-qubit Cliffords, there exists an attack against Protocol~\ref{prot:auth} which successfully passes the verification with non-negligible probability using  $\tau^2 n^{\mc{O}(\log (1/\tau))}$ queries and running in time $\mathsf{poly}(n, 1/\tau) n^{\mc{O}(\log (1/\tau))}$.
\end{theorem}
\begin{proof}
    Consider the attack described in Algorithm~\ref{alg:attack-prot-auth}. The stated complexities can be obtained from \Cref{cor:spt-avg-with-assumption} and the hyperparameter settings stated in Algorithm~\ref{alg:attack-prot-auth}. We now focus on the correctness of the attack. As we learn within average squared error $\tau^2$, and using Jensen's inequality, we see that on average over $\rho(x)$ drawn from $\mc{D}$,
    \begin{equation}
        |h(\rho(x)) - \tr(O\mathcal{E}(\rho(x)))| \leq \tau.
    \end{equation}
    For any query $x$ issued by $\mathcal{V}$, the associated $y$ stored in $D$ was received as the output of a statistical query, implying that
    \begin{equation}
        |y - \tr(O\mathcal{E}(\rho(x)))| \leq \tau.
    \end{equation}
     By triangle inequality, the error between the adversary's prediction $h(\rho(x))$ and the stored $y$ is bounded by $2\tau$. Thus, using Algorithm~\ref{alg:attack-prot-auth}, an adversary is able to efficiently pass Protocol~\ref{prot:auth} with non-negligible probability over the challenge distribution.
\end{proof}

\begin{remark}
   In Protocol \ref{prot:auth}, the parties agree on a query tolerance beforehand. As a consequence, the adversary does not have control of the query tolerance during the attack, necessitating the use of the guarantee of \Cref{cor:spt-avg-with-assumption}, which places no restriction on the tolerance, as opposed to \Cref{thm:spt-avg}. However, when using \Cref{cor:spt-avg-with-assumption}, the output of the CR-QPUF must satisfy Assumption \ref{ass:qpstat-assumption}, i.e., the output should be unbiased. This assumption is satisfied by standard methods for expectation value estimation that could be used to implement the oracle. However, one could add a small amount of biased noise to the output of the CR-QPUF, resulting in our attack failing, and the resulting protocol may be secure. We do not investigate this setting here and leave it for future work. Alternatively, one could also consider a stronger adversarial model, where the adversary can make queries for any tolerance of their choice, in which case the guarantee of \Cref{thm:spt-avg} can be used and Assumption \ref{ass:qpstat-assumption} would not be necessary.
\end{remark}
\begin{remark}
    Naively, one might anticipate that our hardness of learning results would imply a positive security result for CR-QPUFs. However, due to the nature of the verification procedure of the protocol, an adversary only needs to learn to predict approximately correct expectation values, rather than learn the underlying process up to accuracy in diamond distance. This is a significantly easier task, and as such, the hardness results cannot directly be used to prove the security of this protocol. In fact, contrary to this anticipation, we observe the opposite outcome. When the underlying process of the CR-QPUF is chosen from a hard-to-learn ensemble, such as a (approximate) unitary 2-design, then the output expectation values are highly concentrated regardless of the choice of observable and input state, enabling an easy attack with high success probability for multiple rounds. This highlights the importance of the verification method when designing such protocols.
\end{remark} 
\section{Conclusion and Future Work}
\label{sec:outlook}

We have presented a physically well-motivated access model for learning quantum processes. Within this access model, we have studied two important tasks in quantum learning theory, namely shadow tomography of a quantum process as well as learning unitaries with respect to the diamond distance. We have also demonstrated the practical relevance of this access model and our learning algorithm by partially addressing an open question regarding the security of Classical Readout Quantum PUFs.

For shadow tomography of quantum processes, \Cref{alg:avg-spt} succeeds for arbitrary quantum processes, but requires the observables to be local and the distribution over states to be invariant under local Clifford operations. An exciting direction for future work is to consider natural restrictions on the process instead, potentially resulting in efficient algorithms for more general states and observables.

For the case of diamond distance learning, we have shown a lower bound for learning a general class of unitaries, through a reduction from distinguishing it from the depolarizing channel. While this bound has allowed us to show the hardness of learning unitary $2$-designs and the Haar measure, our bound does not hold for non-unitary channels. We believe this bound can be generalized to the non-unitary setting by defining \emph{ancilla-assisted} QPSQs, and leave this for future work.

Furthermore, it is interesting to compare learnability in this framework to the usual quantum statistical query framework for learning states. When considering the task of learning classical functions encoded either as a unitary or as a quantum example, it is interesting to see whether the additional choice of input state available to a learner in our model can provide any advantage, i.e. \textit{are there any separations between QPSQ and QSQ learners when looking at classical functions?}. Another compelling question in this context is whether we can show a meaningful, formal separation between QPSQ and classical learners (beyond what has already been shown through generalizing QSQs).

In terms of applications, there is much to explore in cryptography. We have demonstrated a connection between learning algorithms for shadow process tomography and attacks against CR-QPUFs. As our algorithm has a quasipolynomial query and time complexity, we are unable to conclusively show the vulnerability of CR-QPUFs. However, our approach shows that any new algorithm for this problem with a polynomial complexity would result in an efficient attack against such protocols, motivating further exploration of efficient algorithms for shadow process tomography from QPSQs. Another interesting direction of research would be to instead identify cryptographic schemes whose security can be proven from lower bounds for learning in the QPSQ access model.

Finally, it would be intriguing to observe the implementation of our learner on actual hardware or its application to data acquired from real physical experiments.
\section*{Acknowledgements}
The authors thank Armando Angrisani for his valuable inputs and discussions, especially the discussions towards establishing the definition of QPSQ, and for sharing related results from \cite{angrisani2023learning} with us during the project. We thank Alexander Nietner for pointing out an error in a security result regarding CR-QPUFs in a previous version of this work and other valuable discussions. We also thank Elham Kashefi, Dominik Leichtle, Laura Lewis, Yao Ma, Eliott Mamon and Sean Thrasher for interesting discussions and comments at different stages of this work. We are incredibly grateful to anonymous reviewers for valuable feedback on this manuscript. The authors acknowledge the support of the Quantum Advantage Pathfinder (QAP), with grant reference EP/X026167/1, and the UK Engineering and Physical Sciences Research Council. 

\bibliographystyle{quantum}
\bibliography{mybiblio}

\begin{thebibliography}{10}

\bibitem{valiant1984theory}
Leslie~G Valiant.
\newblock ``A theory of the learnable''.
\newblock \href{https://dx.doi.org/10.1145/1968.1972}{Communications of the ACM {\bf 27}, 1134--1142}~(1984).

\bibitem{SQLearning}
Michael Kearns.
\newblock ``Efficient noise-tolerant learning from statistical queries''.
\newblock \href{https://dx.doi.org/10.1145/293347.293351}{Journal of the ACM (JACM) {\bf 45}, 983--1006}~(1998).

\bibitem{QPAC}
Nader~H. Bshouty and Jeffrey~C. Jackson.
\newblock ``Learning dnf over the uniform distribution using a quantum example oracle''.
\newblock In Proceedings of the Eighth Annual Conference on Computational Learning Theory.
\newblock \href{https://dx.doi.org/10.1145/225298.225312}{Page 118–127}.
\newblock COLT'95. Association for Computing Machinery~(1995).

\bibitem{arunachalam2018optimal}
Srinivasan Arunachalam and Ronald De~Wolf.
\newblock ``Optimal quantum sample complexity of learning algorithms''.
\newblock The Journal of Machine Learning Research {\bf 19}, 2879--2878~(2018).

\bibitem{arunachalam2020quantum}
Srinivasan Arunachalam, Alex~B. Grilo, and Henry Yuen.
\newblock ``Quantum statistical query learning''~(2020).
\newblock  \href{http://arxiv.org/abs/2002.08240}{arxiv:2002.08240}.

\bibitem{arunachalam2023role}
Srinivasan Arunachalam, Vojtech Havlicek, and Louis Schatzki.
\newblock ``On the role of entanglement and statistics in learning''.
\newblock Advances in Neural Information Processing Systems {\bf 36}, 55064--55076~(2024).

\bibitem{hinsche2023one}
M~Hinsche, M~Ioannou, A~Nietner, J~Haferkamp, Y~Quek, D~Hangleiter, J-P Seifert, J~Eisert, and R~Sweke.
\newblock ``One t gate makes distribution learning hard''.
\newblock \href{https://dx.doi.org/10.1103/physrevlett.130.240602}{Physical Review Letters {\bf 130}, 240602}~(2023).

\bibitem{nietner2023average}
Alexander Nietner, Marios Ioannou, Ryan Sweke, Richard Kueng, Jens Eisert, Marcel Hinsche, and Jonas Haferkamp.
\newblock ``On the average-case complexity of learning output distributions of quantum circuits''~(2023).
\newblock  \href{http://arxiv.org/abs/2305.05765}{arxiv:2305.05765}.

\bibitem{nietner2023unifying}
Alexander Nietner.
\newblock ``Unifying (quantum) statistical and parametrized (quantum) algorithms''~(2023).
\newblock  \href{http://arxiv.org/abs/2310.17716}{arxiv:2310.17716}.

\bibitem{atici2005improved}
Alp Atici and Rocco~A Servedio.
\newblock ``Improved bounds on quantum learning algorithms''.
\newblock \href{https://dx.doi.org/10.1007/s11128-005-0001-2}{Quantum Information Processing {\bf 4}, 355--386}~(2005).

\bibitem{grilolearning}
Alex~B. Grilo, Iordanis Kerenidis, and Timo Zijlstra.
\newblock ``Learning-with-errors problem is easy with quantum samples''.
\newblock \href{https://dx.doi.org/10.1103/physreva.99.032314}{Physical Review A {\bf 99}, 032314}~(2019).

\bibitem{bisio2009optimal}
Alessandro Bisio, Giulio Chiribella, Giacomo~Mauro D'Ariano, Stefano Facchini, and Paolo Perinotti.
\newblock ``Optimal quantum tomography''.
\newblock \href{https://dx.doi.org/10.1109/jstqe.2009.2029243}{IEEE Journal of Selected Topics in Quantum Electronics {\bf 15}, 1646--1660}~(2009).

\bibitem{o2016efficient}
Ryan O'Donnell and John Wright.
\newblock ``Efficient quantum tomography''.
\newblock In Proceedings of the forty-eighth annual ACM symposium on Theory of Computing.
\newblock \href{https://dx.doi.org/10.1145/2897518.2897544}{Pages 899--912}.
\newblock ~(2016).

\bibitem{xu2018neural}
Qian Xu and Shuqi Xu.
\newblock ``Neural network state estimation for full quantum state tomography''~(2018).
\newblock  \href{http://arxiv.org/abs/1811.06654}{arxiv:1811.06654}.

\bibitem{classical_shadow_tomography}
Hsin-Yuan Huang, Richard Kueng, and John Preskill.
\newblock ``Predicting many properties of a quantum system from very few measurements''.
\newblock \href{https://dx.doi.org/10.1038/s41567-020-0932-7}{Nature Physics {\bf 16}, 1050--1057}~(2020).

\bibitem{aaronson2018shadow}
Scott Aaronson.
\newblock ``Shadow tomography of quantum states''.
\newblock In Proceedings of the 50th annual ACM SIGACT symposium on theory of computing.
\newblock \href{https://dx.doi.org/10.1145/3188745.3188802}{Pages 325--338}.
\newblock ~(2018).

\bibitem{childs2022quantum}
Andrew~M Childs, Tongyang Li, Jin-Peng Liu, Chunhao Wang, and Ruizhe Zhang.
\newblock ``Quantum algorithms for sampling log-concave distributions and estimating normalizing constants''.
\newblock Advances in Neural Information Processing Systems {\bf 35}, 23205--23217~(2022).

\bibitem{montanaro2017learning}
Ashley Montanaro.
\newblock ``Learning stabilizer states by bell sampling''~(2017).
\newblock  \href{http://arxiv.org/abs/1707.04012}{arxiv:1707.04012}.

\bibitem{mohseni2008quantum}
Masoud Mohseni, Ali~T Rezakhani, and Daniel~A Lidar.
\newblock ``Quantum-process tomography: Resource analysis of different strategies''.
\newblock \href{https://dx.doi.org/10.1103/physreva.77.032322}{Physical Review A {\bf 77}, 032322}~(2008).

\bibitem{chung2018sample}
Kai-Min Chung and Han-Hsuan Lin.
\newblock ``Sample efficient algorithms for learning quantum channels in pac model and the approximate state discrimination problem''.
\newblock In 16th Conference on the Theory of Quantum Computation, Communication and Cryptography (TQC 2021).
\newblock \href{https://dx.doi.org/10.4230/LIPIcs.TQC.2021.3}{Pages 3:1--3:22}.
\newblock Schloss Dagstuhl--Leibniz-Zentrum f{\"u}r Informatik~(2021).

\bibitem{haah2023query}
Jeongwan Haah, Robin Kothari, Ryan O’Donnell, and Ewin Tang.
\newblock ``Query-optimal estimation of unitary channels in diamond distance''.
\newblock In 2023 IEEE 64th Annual Symposium on Foundations of Computer Science (FOCS).
\newblock \href{https://dx.doi.org/10.1109/focs57990.2023.00028}{Pages 363--390}.
\newblock IEEE~(2023).

\bibitem{arunachalam2017guest}
Srinivasan Arunachalam and Ronald de~Wolf.
\newblock ``Guest column: A survey of quantum learning theory''.
\newblock \href{https://dx.doi.org/10.1145/3106700.3106710}{ACM Sigact News {\bf 48}, 41--67}~(2017).

\bibitem{wiebe2014hamiltonian}
Nathan Wiebe, Christopher Granade, Christopher Ferrie, and David~G Cory.
\newblock ``Hamiltonian learning and certification using quantum resources''.
\newblock \href{https://dx.doi.org/10.1103/physrevlett.112.190501}{Physical review letters {\bf 112}, 190501}~(2014).

\bibitem{grimsley2019adaptive}
Harper~R Grimsley, Sophia~E Economou, Edwin Barnes, and Nicholas~J Mayhall.
\newblock ``An adaptive variational algorithm for exact molecular simulations on a quantum computer''.
\newblock \href{https://dx.doi.org/10.1038/s41467-019-10988-2}{Nature communications {\bf 10}, 3007}~(2019).

\bibitem{scott2008optimizing}
Andrew~James Scott.
\newblock ``Optimizing quantum process tomography with unitary 2-designs''.
\newblock \href{https://dx.doi.org/10.1088/1751-8113/41/5/055308}{Journal of Physics A: Mathematical and Theoretical {\bf 41}, 055308}~(2008).

\bibitem{levy2021classical}
Ryan Levy, Di~Luo, and Bryan~K Clark.
\newblock ``Classical shadows for quantum process tomography on near-term quantum computers''.
\newblock \href{https://dx.doi.org/10.1103/physrevresearch.6.013029}{Physical Review Research {\bf 6}, 013029}~(2024).

\bibitem{huang2022foundations}
Hsin-Yuan Huang, Steven~T Flammia, and John Preskill.
\newblock ``Foundations for learning from noisy quantum experiments''~(2022).
\newblock  \href{http://arxiv.org/abs/2204.13691}{arxiv:2204.13691}.

\bibitem{blume2017demonstration}
Robin Blume-Kohout, John~King Gamble, Erik Nielsen, Kenneth Rudinger, Jonathan Mizrahi, Kevin Fortier, and Peter Maunz.
\newblock ``Demonstration of qubit operations below a rigorous fault tolerance threshold with gate set tomography''.
\newblock \href{https://dx.doi.org/10.1038/ncomms14485}{Nature communications {\bf 8}, 14485}~(2017).

\bibitem{harper2020efficient}
Robin Harper, Steven~T Flammia, and Joel~J Wallman.
\newblock ``Efficient learning of quantum noise''.
\newblock \href{https://dx.doi.org/10.1038/s41567-020-0992-8}{Nature Physics {\bf 16}, 1184--1188}~(2020).

\bibitem{strikis2021learning}
Armands Strikis, Dayue Qin, Yanzhu Chen, Simon~C Benjamin, and Ying Li.
\newblock ``Learning-based quantum error mitigation''.
\newblock \href{https://dx.doi.org/10.1103/prxquantum.2.040330}{PRX Quantum {\bf 2}, 040330}~(2021).

\bibitem{em}
Yihui Quek, Daniel Stilck~Fran{\c{c}}a, Sumeet Khatri, Johannes~Jakob Meyer, and Jens Eisert.
\newblock ``Exponentially tighter bounds on limitations of quantum error mitigation''.
\newblock \href{https://dx.doi.org/10.1038/s41567-024-02536-7}{Nature Physics {\bf 20}, 1648--1658}~(2024).

\bibitem{schuld2015introduction}
Maria Schuld, Ilya Sinayskiy, and Francesco Petruccione.
\newblock ``An introduction to quantum machine learning''.
\newblock \href{https://dx.doi.org/10.1080/00107514.2014.964942}{Contemporary Physics {\bf 56}, 172--185}~(2015).

\bibitem{boneh2013quantum}
Dan Boneh and Mark Zhandry.
\newblock ``Quantum-secure message authentication codes''.
\newblock In Thomas Johansson and Phong~Q. Nguyen, editors, Advances in Cryptology -- EUROCRYPT 2013.
\newblock \href{https://dx.doi.org/10.1007/978-3-642-38348-9_35}{Pages 592--608}.
\newblock Berlin, Heidelberg~(2013). Springer Berlin Heidelberg.

\bibitem{kaplan2016breaking}
Marc Kaplan, Ga{\"e}tan Leurent, Anthony Leverrier, and Mar{\'i}a Naya-Plasencia.
\newblock ``Breaking symmetric cryptosystems using quantum period finding''.
\newblock In Matthew Robshaw and Jonathan Katz, editors, Advances in Cryptology -- CRYPTO 2016.
\newblock \href{https://dx.doi.org/10.1007/978-3-662-53008-5_8}{Pages 207--237}.
\newblock Berlin, Heidelberg~(2016). Springer Berlin Heidelberg.

\bibitem{santoli2016using}
Thomas Santoli and Christian Schaffner.
\newblock ``Using simon's algorithm to attack symmetric-key cryptographic primitives''.
\newblock \href{https://dx.doi.org/10.26421/qic17.1-2-4}{Quantum Information \& Computation {\bf 17}, 65--78}~(2017).

\bibitem{chevalier2022security}
C{\'e}line Chevalier, Ehsan Ebrahimi, and Quoc-Huy Vu.
\newblock ``On security notions for encryption in a quantum world''.
\newblock In International Conference on Cryptology in India.
\newblock \href{https://dx.doi.org/10.1007/978-3-031-22912-1_26}{Pages 592--613}.
\newblock Springer~(2022).

\bibitem{QPUF}
Myrto Arapinis, Mahshid Delavar, Mina Doosti, and Elham Kashefi.
\newblock ``Quantum physical unclonable functions: Possibilities and impossibilities''.
\newblock \href{https://dx.doi.org/10.22331/q-2021-06-15-475}{Quantum {\bf 5}, 475}~(2021).

\bibitem{CRQPUF-Original}
Koustubh Phalak, Abdullah Ash-Saki, Mahabubul Alam, Rasit~Onur Topaloglu, and Swaroop Ghosh.
\newblock ``Quantum {PUF} for security and trust in quantum computing''.
\newblock \href{https://dx.doi.org/10.1109/jetcas.2021.3077024}{{IEEE} Journal on Emerging and Selected Topics in Circuits and Systems {\bf 11}, 333--342}~(2021).

\bibitem{CR-QPUF-single}
Niklas Pirnay, Anna Pappa, and Jean-Pierre Seifert.
\newblock ``Learning classical readout quantum {PUFs} based on single-qubit gates''.
\newblock \href{https://dx.doi.org/10.1007/s42484-022-00073-1}{Quantum Machine Intelligence {\bf 4}, 14}~(2022).

\bibitem{arb_proc}
Hsin-Yuan Huang, Sitan Chen, and John Preskill.
\newblock ``Learning to predict arbitrary quantum processes''.
\newblock \href{https://dx.doi.org/10.1103/prxquantum.4.040337}{PRX Quantum {\bf 4}, 040337}~(2023).

\bibitem{montanaro2008quantum}
Ashley Montanaro and Tobias~J Osborne.
\newblock ``Quantum boolean functions''~(2008).
\newblock  \href{http://arxiv.org/abs/0810.2435}{arxiv:0810.2435}.

\bibitem{fanizza2022learning}
Marco Fanizza, Yihui Quek, and Matteo Rosati.
\newblock ``Learning quantum processes without input control''.
\newblock \href{https://dx.doi.org/10.1103/prxquantum.5.020367}{PRX Quantum {\bf 5}, 020367}~(2024).

\bibitem{caro2023classical}
Matthias~C Caro, Marcel Hinsche, Marios Ioannou, Alexander Nietner, and Ryan Sweke.
\newblock ``Classical verification of quantum learning''.
\newblock In 15th Innovations in Theoretical Computer Science Conference (ITCS 2024).
\newblock \href{https://dx.doi.org/10.4230/LIPIcs.ITCS.2024.24}{Pages 24:1--24:23}.
\newblock Schloss Dagstuhl--Leibniz-Zentrum f{\"u}r Informatik~(2024).

\bibitem{feldman2017general}
Vitaly Feldman.
\newblock ``A general characterization of the statistical query complexity''.
\newblock In Conference on learning theory.
\newblock Pages 785--830.
\newblock PMLR~(2017).

\bibitem{angrisani2023learning}
Armando Angrisani.
\newblock ``Learning unitaries with quantum statistical queries''~(2023).
\newblock  \href{http://arxiv.org/abs/2310.02254}{arxiv:2310.02254}.

\bibitem{chen2022exponential}
Sitan Chen, Jordan Cotler, Hsin-Yuan Huang, and Jerry Li.
\newblock ``Exponential separations between learning with and without quantum memory''.
\newblock In 2021 IEEE 62nd Annual Symposium on Foundations of Computer Science (FOCS).
\newblock \href{https://dx.doi.org/10.1109/focs52979.2021.00063}{Pages 574--585}.
\newblock IEEE~(2022).

\bibitem{kunjummen2023shadow}
Jonathan Kunjummen, Minh~C Tran, Daniel Carney, and Jacob~M Taylor.
\newblock ``Shadow process tomography of quantum channels''.
\newblock \href{https://dx.doi.org/10.1103/physreva.107.042403}{Physical Review A {\bf 107}, 042403}~(2023).

\bibitem{caro2022learning}
Matthias~C Caro.
\newblock ``Learning quantum processes and hamiltonians via the pauli transfer matrix''.
\newblock \href{https://dx.doi.org/10.1145/3670418}{ACM Transactions on Quantum Computing {\bf 5}, 1--53}~(2024).

\bibitem{nadimpalli2024pauli}
Shivam Nadimpalli, Natalie Parham, Francisca Vasconcelos, and Henry Yuen.
\newblock ``On the pauli spectrum of qac0''.
\newblock In Proceedings of the 56th Annual ACM Symposium on Theory of Computing.
\newblock \href{https://dx.doi.org/10.1145/3618260.3649662}{Pages 1498--1506}.
\newblock ~(2024).

\bibitem{nielsen2002quantum}
Michael~A Nielsen and Isaac~L Chuang.
\newblock ``Quantum computation and quantum information''.
\newblock Cambridge university press. ~(2010).

\bibitem{mele2023introduction}
Antonio~Anna Mele.
\newblock ``Introduction to haar measure tools in quantum information: A beginner's tutorial''.
\newblock \href{https://dx.doi.org/10.22331/q-2024-05-08-1340}{Quantum {\bf 8}, 1340}~(2024).

\bibitem{random_measurement}
Andreas Elben, Steven~T. Flammia, Hsin-Yuan Huang, Richard Kueng, John Preskill, Beno{\^{\i}}t Vermersch, and Peter Zoller.
\newblock ``The randomized measurement toolbox''.
\newblock \href{https://dx.doi.org/10.1038/s42254-022-00535-2}{Nature Reviews Physics {\bf 5}, 9--24}~(2022).

\bibitem{huang2022provably}
Hsin-Yuan Huang, Richard Kueng, Giacomo Torlai, Victor~V Albert, and John Preskill.
\newblock ``Provably efficient machine learning for quantum many-body problems''.
\newblock \href{https://dx.doi.org/10.1126/science.abk3333}{Science {\bf 377}, eabk3333}~(2022).

\bibitem{kanade2019a}
V~Kanade, A~Rocchetto, and S~Severini.
\newblock ``Learning dnfs under product distributions via $\mu$-biased quantum fourier sampling''.
\newblock \href{https://dx.doi.org/10.26421/QIC19.15-16-1}{Quantum Information and Computation {\bf 19}, 1261--1278}~(2019).

\bibitem{Caro_2020}
Matthias~C. Caro.
\newblock ``Quantum learning boolean linear functions w.r.t. product distributions''.
\newblock \href{https://dx.doi.org/10.1007/s11128-020-02661-1}{Quantum Information Processing {\bf 19}, 172}~(2020).

\bibitem{Du_2021}
Yuxuan Du, Min-Hsiu Hsieh, Tongliang Liu, Dacheng Tao, and Nana Liu.
\newblock ``Quantum noise protects quantum classifiers against adversaries''.
\newblock \href{https://dx.doi.org/10.1103/physrevresearch.3.023153}{Physical Review Research {\bf 3}, 023153}~(2021).

\bibitem{gollakota2022hardness}
Aravind Gollakota and Daniel Liang.
\newblock ``On the hardness of pac-learning stabilizer states with noise''.
\newblock \href{https://dx.doi.org/10.22331/q-2022-02-02-640}{Quantum {\bf 6}, 640}~(2022).

\bibitem{kim2020quantum}
Changjun Kim, Kyungdeock~Daniel Park, and June-Koo Rhee.
\newblock ``Quantum error mitigation with artificial neural network''.
\newblock \href{https://dx.doi.org/10.1109/access.2020.3031607}{IEEE Access {\bf 8}, 188853--188860}~(2020).

\bibitem{huang2021information}
Hsin-Yuan Huang, Richard Kueng, and John Preskill.
\newblock ``Information-theoretic bounds on quantum advantage in machine learning''.
\newblock \href{https://dx.doi.org/10.1103/physrevlett.126.190505}{Physical Review Letters {\bf 126}, 190505}~(2021).

\bibitem{hadfield2020measurements_locally_biased}
Charles Hadfield, Sergey Bravyi, Rudy Raymond, and Antonio Mezzacapo.
\newblock ``Measurements of quantum hamiltonians with locally-biased classical shadows''.
\newblock \href{https://dx.doi.org/10.1007/s00220-022-04343-8}{Communications in Mathematical Physics {\bf 391}, 951--967}~(2022).

\bibitem{Huang_2021_derandom}
Hsin-Yuan Huang, Richard Kueng, and John Preskill.
\newblock ``Efficient estimation of pauli observables by derandomization''.
\newblock \href{https://dx.doi.org/10.1103/physrevlett.127.030503}{Physical Review Letters {\bf 127}, 030503}~(2021).

\bibitem{Wu2023overlappedgrouping}
Bujiao Wu, Jinzhao Sun, Qi~Huang, and Xiao Yuan.
\newblock ``Overlapped grouping measurement: {A} unified framework for measuring quantum states''.
\newblock \href{https://dx.doi.org/10.22331/q-2023-01-13-896}{{Quantum} {\bf 7}, 896}~(2023).

\bibitem{webb2015clifford}
Zak Webb.
\newblock ``The clifford group forms a unitary 3-design''~(2015).
\newblock  \href{http://arxiv.org/abs/1510.02769}{arxiv:1510.02769}.

\bibitem{zhu2017multiqubit}
Huangjun Zhu.
\newblock ``Multiqubit clifford groups are unitary 3-designs''.
\newblock \href{https://dx.doi.org/10.1103/physreva.96.062336}{Physical Review A {\bf 96}, 062336}~(2017).

\bibitem{low2009learning}
Richard~A Low.
\newblock ``Learning and testing algorithms for the clifford group''.
\newblock \href{https://dx.doi.org/10.1103/physreva.80.052314}{Physical Review A—Atomic, Molecular, and Optical Physics {\bf 80}, 052314}~(2009).

\bibitem{lai2022learning}
Ching-Yi Lai and Hao-Chung Cheng.
\newblock ``Learning quantum circuits of some t gates''.
\newblock \href{https://dx.doi.org/10.1109/tit.2022.3151760}{IEEE Transactions on Information Theory {\bf 68}, 3951--3964}~(2022).

\bibitem{Haferkamp_2021}
Jonas Haferkamp and Nicholas Hunter-Jones.
\newblock ``Improved spectral gaps for random quantum circuits: Large local dimensions and all-to-all interactions''.
\newblock \href{https://dx.doi.org/10.1103/physreva.104.022417}{Physical Review A {\bf 104}, 022417}~(2021).

\bibitem{meckes2013spectral}
Elizabeth Meckes and Mark Meckes.
\newblock ``{Spectral measures of powers of random matrices}''.
\newblock \href{https://dx.doi.org/10.1214/ECP.v18-2551}{Electronic Communications in Probability {\bf 18}, 1 -- 13}~(2013).

\bibitem{baumgartner2011inequality}
Bernhard Baumgartner.
\newblock ``An inequality for the trace of matrix products, using absolute values''~(2011).
\newblock  \href{http://arxiv.org/abs/1106.6189}{arxiv:1106.6189}.

\bibitem{ruhrmair2014pufs}
Ulrich R{\"u}hrmair and Daniel~E Holcomb.
\newblock ``Pufs at a glance''.
\newblock In the conference on Design, Automation \& Test in Europe.
\newblock \href{https://dx.doi.org/10.7873/date.2014.360}{Page 347}.
\newblock European Design and Automation Association~(2014).

\bibitem{chang2017retrospective}
Chip-Hong Chang, Yue Zheng, and Le~Zhang.
\newblock ``A retrospective and a look forward: Fifteen years of physical unclonable function advancement''.
\newblock \href{https://dx.doi.org/10.1109/mcas.2017.2713305}{IEEE Circuits and Systems Magazine {\bf 17}, 32--62}~(2017).

\bibitem{Id-QPUF}
Mina Doosti, Niraj Kumar, Mahshid Delavar, and Elham Kashefi.
\newblock ``Client-server identification protocols with quantum {PUF}''.
\newblock \href{https://dx.doi.org/10.1145/3484197}{{ACM} Transactions on Quantum Computing {\bf 2}, 1--40}~(2021).

\bibitem{delavar2017puf}
Mahshid Delavar, Sattar Mirzakuchaki, Mohammad~Hassan Ameri, and Javad Mohajeri.
\newblock ``Puf-based solutions for secure communications in advanced metering infrastructure (ami)''.
\newblock \href{https://dx.doi.org/10.1002/dac.3195}{International Journal of Communication Systems {\bf 30}, e3195}~(2017).

\bibitem{pappu2002physical}
Ravikanth Pappu, Ben Recht, Jason Taylor, and Neil Gershenfeld.
\newblock ``Physical one-way functions''.
\newblock \href{https://dx.doi.org/10.1126/science.1074376}{Science {\bf 297}, 2026--2030}~(2002).

\bibitem{guajardo2007fpga}
Jorge Guajardo, Sandeep~S Kumar, Geert-Jan Schrijen, and Pim Tuyls.
\newblock ``Fpga intrinsic pufs and their use for ip protection''.
\newblock In International workshop on cryptographic hardware and embedded systems.
\newblock \href{https://dx.doi.org/10.1007/978-3-540-74735-2_5}{Pages 63--80}.
\newblock Springer~(2007).

\bibitem{gassend2002silicon}
Blaise Gassend, Dwaine Clarke, Marten Van~Dijk, and Srinivas Devadas.
\newblock ``Silicon physical random functions''.
\newblock In 9th ACM conference on Computer and communications security.
\newblock \href{https://dx.doi.org/10.1145/586110.586132}{Pages 148--160}.
\newblock ACM~(2002).

\bibitem{suh2007physical}
G~Edward Suh and Srinivas Devadas.
\newblock ``Physical unclonable functions for device authentication and secret key generation''.
\newblock In 44th ACM/IEEE Design Automation Conference.
\newblock \href{https://dx.doi.org/10.1145/1278480.1278484}{Pages 9--14}.
\newblock IEEE~(2007).

\bibitem{ruhrmair2010modeling}
Ulrich R{\"u}hrmair, Frank Sehnke, Jan S{\"o}lter, Gideon Dror, Srinivas Devadas, and J{\"u}rgen Schmidhuber.
\newblock ``Modeling attacks on physical unclonable functions''.
\newblock In Proceedings of the 17th ACM conference on Computer and communications security.
\newblock \href{https://dx.doi.org/10.1145/1866307.1866335}{Pages 237--249}.
\newblock ~(2010).

\bibitem{ganji2016strong}
Fatemeh Ganji, Shahin Tajik, Fabian F{\"a}{\ss}ler, and Jean-Pierre Seifert.
\newblock ``Strong machine learning attack against {PUF}s with no mathematical model''.
\newblock In International Conference on Cryptographic Hardware and Embedded Systems.
\newblock \href{https://dx.doi.org/10.1007/978-3-662-53140-2_19}{Pages 391--411}.
\newblock Springer~(2016).

\bibitem{tebelmann2019side}
Lars Tebelmann, Michael Pehl, and Vincent Immler.
\newblock ``Side-channel analysis of the tero {PUF}''.
\newblock In International Workshop on Constructive Side-Channel Analysis and Secure Design.
\newblock \href{https://dx.doi.org/10.1007/978-3-030-16350-1_4}{Pages 43--60}.
\newblock Springer~(2019).

\bibitem{khalafalla2019pufs}
Mahmoud Khalafalla and Catherine Gebotys.
\newblock ``{PUF}s deep attacks: Enhanced modeling attacks using deep learning techniques to break the security of double arbiter {PUF}s''.
\newblock In 2019 Design, Automation \& Test in Europe Conference \& Exhibition (DATE).
\newblock \href{https://dx.doi.org/10.23919/date.2019.8714862}{Pages 204--209}.
\newblock IEEE~(2019).

\bibitem{doosti2021unified}
Mina Doosti, Mahshid Delavar, Elham Kashefi, and Myrto Arapinis.
\newblock ``A unified framework for quantum unforgeability''~(2021).
\newblock  \href{http://arxiv.org/abs/2103.13994}{arxiv:2103.13994}.

\bibitem{alagic2018unforgeable}
Gorjan Alagic, Tommaso Gagliardoni, and Christian Majenz.
\newblock ``Unforgeable quantum encryption''.
\newblock In Jesper~Buus Nielsen and Vincent Rijmen, editors, Advances in Cryptology -- EUROCRYPT 2018.
\newblock \href{https://dx.doi.org/10.1007/978-3-319-78372-7_16}{Pages 489--519}.
\newblock Cham~(2018). Springer International Publishing.

\bibitem{observable-measurement-2}
Ophelia Crawford, Barnaby van Straaten, Daochen Wang, Thomas Parks, Earl Campbell, and Stephen Brierley.
\newblock ``Efficient quantum measurement of pauli operators in the presence of finite sampling error''.
\newblock \href{https://dx.doi.org/10.22331/q-2021-01-20-385}{Quantum {\bf 5}, 385}~(2021).

\end{thebibliography}
\end{document}